%% file: ArxivFinal.tex
\documentclass{article}

\input{preamble}

\title{The spin-statistics theorem for topological quantum field theories}
\author{Luuk Stehouwer}
\date{\today}

\begin{document}

\maketitle

\begin{abstract}
We establish the spin-statistics theorem for topological quantum field theories (TQFTs) in the framework of Atiyah. 
We incorporate spin via spin structures on bordisms, and represent statistics using super vector spaces. 
Unitarity is implemented using dagger categories, in a manner that is equivalent to the approach of Freed-Hopkins, who employed $\mathbb{Z}/2$-equivariant functors to address reflection-positivity.
A key contribution of our work is the introduction of the notion of fermionically dagger compact categories, which extends the well-established concept of dagger compact categories. 
We show that both the spin bordism category and the category of super Hilbert spaces are examples of fermionically dagger compact categories. 
The spin-statistics theorem for TQFTs emerges as a specific case of a more general result concerning symmetric monoidal dagger functors between fermionically dagger compact categories. 
\end{abstract}

\tableofcontents

\newpage

\section{Introduction}

 \epigraph{``[Spin and statistics] must necessarily go together, but we have not been able
to find a way of reproducing his [(Pauli's)] arguments on an elementary level."}{Richard Feynman}

The Spin-Statistics Theorem, an essential cornerstone in quantum field theory (QFT), relates the spin of elementary particles to their statistical behavior. 
This theorem asserts that particles with integer spin conform to Bose-Einstein statistics, while those with half-integer spin follow Fermi-Dirac statistics. 
Beyond its theoretical elegance, the fact that the floor on which we are standing right now does not collapse into a condensate, is due to the fact that matter has half-integer spin, applying the Pauli exclusion principle.
Ever since the first proof of the Spin-Statistics Theorem in QFT by Wolfgang Pauli in 1940 \cite{paulispinstatistics}, mathematical physicists have been seeking proofs that are mathematically rigorous and emphasize the exact necessary assumptions.
Even though the Spin-Statistics Theorem has been a textbook account for decades \cite{streaterwightman, weinberg1995quantum}, mathematical physicists continued to find new proofs in several axiomatic frameworks for quantum field theory, such as the Wightman axioms \cite{burgoynespinstatistics, deligne1999quantum} and in algebraic QFT \cite{guido1995algebraicspinstatistics}.
On the other hand, even recent theoretical physics literature is still discussing the possibility of providing more conceptually insightful proofs \cite{duck1998toward, santamatospinstatistics}.

In this context, our work employs category-theoretical language to provide a proof of the Spin-Statistics Theorem for topological quantum field theories (TQFTs). 
Our contribution aligns with the recent trend of applying category theory to quantum physics, not only in TQFTs, but also in non-invertible symmetries of non-topological QFTs and topological phases of matter.
We strongly rely on the work of Freed and Hopkins on reflection-positivity, who demonstrated the theorem's validity in the realm of invertible topological field theories \cite{freedhopkins}.
One main insight of our work is to encapsulate the distinct properties of categories describing fermionic systems, such as super Hilbert spaces and spin bordism, which are necessary to make Spin-Statistics a theorem. 
This is the concept of ``fermionically dagger compact" categories, a novel extension of the well-known notion of dagger compact categories relevant in quantum information theory \cite{abramsky2004categorical, selinger2007positive}, see Definition \ref{def:fermdagger}.
We hope our approach enriches the mathematical physicist's understanding of the Spin-Statistics Theorem and inspires further research at the intersection of category theory and QFT.

\subsection*{Unitary TQFT and dagger categories}

It has long been understood that unitarity is a crucial assumption in the proof of the Spin-Statistics Theorem.\footnote{In the rest of this paper, we will use the term `unitary' to refer to both Lorentzian signature unitary QFTs as well as reflection-positive Euclidean QFTs. On the one hand, this is justified, because a Lorentzian QFT is unitary if and only if its Wick-rotated reflection-positive QFT is reflection-positive. On the other hand it ignores the fact that this relationship between the reflection-positivity axiom and unitarity is non-obvious.}
Several authors \cite{atiyahTFT, baezdagger} \cite[Appendix G]{turaev} have suggested that unitarity should be implemented in functorial field theory by requiring the theory to be a dagger functor. \footnote{The formulation using $\Z/2$-equivariant functors by \cite{freedhopkins} can be shown to be equivalent to this approach \cite[Theorem 6.2.7]{mythesis}.} Here a dagger category is a category $\mathcal{D}$ together with an involution functor $\dagger\colon \mathcal{D} \to \mathcal{D}^{\op}$, which is the identity on objects.
We refer to Sections \ref{sec:Jan} and \ref{sec:symmon} for preliminaries on dagger categories and symmetric monoidal dagger categories, respectively.
A \emph{bosonic unitary $d$-dimensional TQFT} is thus a symmetric monoidal dagger functor
\[
 \Bord^{SO}_{d} \to \Hilb
\]
from the symmetric monoidal dagger oriented bordism category\footnote{We will not go into details on how to make bordism categories into symmetric monoidal dagger categories in this paper, instead referring to \cite{mythesis,higherdagger}.} to the symmetric monoidal dagger category of Hilbert spaces.
This definition in particular implies that the partition function on the orientation reversed manifold is the complex conjugate, a fact familiar to physicists.

Including fermions in this definition is straightforward:
state spaces should be $\Z/2$-graded by the fermion parity operator $(-1)^F$. 
The crucial property that distinguishes odd from even states in this grading is their statistics: bosonic states have Bose statistics and fermionic states have Fermi statistics. 
This is implemented by requiring the appropriate exchange operations
\[
\phi_1 \phi_2 = \phi_2 \phi_1 \qquad \psi_1 \psi_2 = - \psi_2 \psi_1,
\]
for bosons $\phi_1,\phi_2$ and fermions $\psi_1,\psi_2$.
Mathematically, we have to equip the monoidal category of $\Z/2$-graded vector spaces with the braiding that implements the Koszul sign rule. 
Therefore we have to replace Hilbert spaces by super Hilbert spaces $\sHilb$, see Section \ref{sec:shilb} for our conventions.
The reader is referred to Appendix \ref{app:boringsign} for a detailed analysis comparing alternative conventions, in particular Theorem \ref{th:sHilbvsZ2grhilb}.
On the geometric side, we need to include spin structures on our spacetime bordisms in order to talk about spinors.
We arrive at the following definition.

\begin{definition}
A \emph{fermionic unitary $d$-dimensional TQFT} is a symmetric monoidal dagger functor
\[
 Z\colon \Bord^{\Spin}_{d} \to \sHilb.
\]
\end{definition}

\subsection*{Spin and statistics for TQFTs}

In order to formulate the connection between spin and statistics in the context of TQFTs, let us consider a time slice $Y^{d-1}$, i.e. an object of $\Bord^{\Spin}_d$.
The state space $Z(Y)$ of a fermionic unitary TQFT $Z$ comes equipped with two $\Z/2$-gradings. 
The obvious grading is given by the fermion parity operation $(-1)_{Z(Y)}^F$, the supergrading involution which maps an odd state $\psi$ to $-\psi$ and and even states $\phi$ to themselves.
The other grading is the involution $Z((-1)^{2s}_Y)$ induced by the spin flip $(-1)^{2s}_Y$ on the spin manifold $Y$. 
In more detail, note that the spin group has a canonical element $(-1)^{2s} \in \Spin(d)$, the nontrivial element in the kernel to $SO(d)$.
Since this element is central, acting with it on the $\Spin(d)$-principal bundle induces an automorphism of the spin structure on $Y$.
It can be thought of as rotating particles by $360^\circ$.
This in turn induces a mapping cylinder bordism $(-1)^{2s}_Y\colon Y \to Y$, which as a manifold with boundary is simply $Y \times [0,1]$, but the identification with $Y$ is twisted on one side with the above involution.

We will now argue that the spin-statistics connection should say that $(-1)_{Z(Y)}^F = Z((-1)^{2s}_Y)$.
For this, first recall the usual classification of irreducible representations of $\Spin(3) = SU(2)$ in terms of spin in the special case $d = 3$.
In that case, $(-1)^{2s}$ acts trivially on integer spin representations and nontrivially on half-integer spin representations.
For more general $d$, we no longer have a classification of irreducible representations of $\Spin(d)$ in terms of spin.
However, the spin group still connects to the concept of spin, because in $d$-dimensional Euclidean signature quantum field theory, spinful particles transform in irreducible representations of $\Spin(d)$.
Moreover, since $(-1)^{2s}$ is central, it still acts by $\pm 1$ and so we can still define integer and half-integer spin particles by how $(-1)^{2s}$ acts in the representation.
One consequence of this viewpoint is that a particle has integer spin if and only if its spin representation lifts to $SO(d)$.
We conclude that $(-1)_{Z(Y)}^F = Z((-1)^{2s}_Y)$ is equivalent to saying ``a particle has half-integer spin if and only if it satisfies Fermi statistics and a particle has integer spin if and only if it satisfies Bose statistics".

The first step to translate this formulation of the spin-statistics theorem to a more categorical language is to realize that the spin-statistics connection expresses equivariance with respect to the $2$-group $B\Z/2$, as observed in \cite{TJFspinstatistics}, see Definition \ref{def:antiinvBAaction}.
Both $\sHilb$ and $\Bord_{d}^{\Spin}$ come equipped with $B\Z/2$-actions as described above.
A fermionic TQFT has a spin-statistics connection if and only if it is $B\Z/2$-equivariant:
\[
\begin{tikzcd}
    \arrow[loop, "{\text{ Spin } B\Z/2}"', "Y \mapsto (-1)^{2s}_Y", outer sep=5pt, distance=3.8cm] \Bord^{\Spin}_{d} \ar[rr,"Z"] && \sHilb \arrow[loop, "{\text{ Statistics } B\Z/2}"', "V \mapsto (-1)^F_V", outer sep=4pt, distance=4cm].
\end{tikzcd}
\]

\subsection*{Main results}

With the above explanation in mind, the following deserves to be called `the spin-statistics theorem for TQFTs'.

\begin{reptheorem}{cor:spinstatistics}
    Every unitary fermionic TQFT is $B\Z/2$-equivariant.    
\end{reptheorem}

Our main insight to prove this result, is to isolate a crucial property that categories of fermionic nature such as $\Bord^{\Spin}_d$ and $\sHilb$ share.
Roughly speaking, if $\mathcal{H} = \mathcal{H}_0 \oplus \mathcal{H}_1 \in \sHilb$ is a super Hilbert space, then the `categorically canonical' inner product on $\mathcal{H}^*$ will be positive definite on $\mathcal{H}_0$ but negative definite on $\mathcal{H}_1$, see Corollary \ref{cor:shilbfermdaggercpt}.
In more mathematical details, we first show that if $\mathcal{D}$ is a symmetric monoidal dagger category, there is a symmetric monoidal dagger category $\Herm \mathcal{D}$ in which $\mathcal{D}$ fully faithfully embeds, see Definition \ref{def:Hermitian completion} and Example \ref{ex:notionofposofdaggercat}.
If $\mathcal{D}$ has duals, we show that $\Herm \mathcal{D}$ comes equipped with a canonical dual functor which is a symmetric monoidal dagger functor, see Definition \ref{def:standardUDF}.
Many dagger categories of `bosonic nature' are dagger compact: we show that a symmetric monoidal dagger category is \emph{dagger compact} in the sense of \cite{selinger2007positive} if this standard dual functor restricts to $\mathcal{D}$, see Definition \ref{def:moderndaggercpt} and Proposition \ref{prop:daggercptTFAE}.
If $\mathcal{D}$ additionally comes equipped with a unitary $B\Z/2$-action, we call it \emph{fermionically dagger compact} if the standard dual functor lands only in $\mathcal{D}$ after twisting it by the $B\Z/2$-action (Definition \ref{def:fermdagger}).
Our main result on fermionically dagger compact categories which implies the spin-statistics theorem is:

\begin{reptheorem}{mainth}
    Let $F\colon \mathcal{D}_1 \to \mathcal{D}_2$ be a symmetric monoidal dagger functor between fermionically dagger compact categories.
    Suppose that for every isomorphism $f\colon x \to y$ in $\mathcal{D}_2$ such that $f^\dagger f$ is an involution, we have that $f^\dagger f$ is the identity.
    Then $F$ is $B\Z/2$-equivariant.
\end{reptheorem}

The main idea of the proof is to first show that a symmetric monoidal dagger functor intertwines the respective standard dagger dual functors. 
Since the standard dagger functors are related to the $B\Z/2$-actions through fermionic dagger compactness, we can derive a relationship between the $B\Z/2$-actions on $\mathcal{D}_1$ and $\mathcal{D}_2$.

\input{acknowledge}

\section{Dagger categories}

\subsection{Anti-involutions and positivity}
\label{sec:Jan}

The goal of this section is to review joint work \cite{jandagger} with Jan Steinebrunner on the relationship between dagger categories and anti-involutive categories using Hermitian pairings.
We start by reviewing the basic definitions of the theory of $\dagger$-categories:

\begin{definition} \label{def:daggerbasics}
A \emph{dagger category} or \emph{$\dagger$-category} is a category $\mathcal{D}$ equipped with a contravariant strict involution $\dagger\colon \mathcal{D} \to \mathcal{D}^{\op}$ which is the identity on objects.
An endomorphism $f\colon x \to x$ in a dagger category is called \emph{self-adjoint} if $f^\dagger = f$.
An isomorphism $f\colon x_1 \to x_2$ is \emph{unitary} if $f^\dagger = f^{-1}$.
A morphism $f \colon x_1 \to x_2$ is \emph{isometric} if $f^\dagger f = \id_{x_1}$.
\end{definition}

We denote by $\pi_0(\mathcal{D})$ the set of isomorphism classes of objects of $\mathcal{D}$.
We denote by $\pi_0^U(\mathcal{D})$ the unitary isomorphism classes, i.e.\ objects of $\mathcal{D}$ modulo the equivalence relation given by unitary isomorphism.

\begin{example}
    The category $\Hilb$ of finite-dimensional complex Hilbert spaces is a dagger category in which the functor $\dagger$ is given by the adjoint of a linear map.
    The notions of unitary and self-adjoint morphisms agree with the familiar notions.
    Since every finite-dimensional vector space admits a Hilbert space structure and any two such are isometrically isomorphic, we have that $\pi_0^U (\Hilb) \cong \pi_0 (\Hilb) \cong \mathbb{N}$ given by the dimension.
    \footnote{Infinite-dimensional Hilbert spaces also form a dagger category. Because we focus on topological field theory, we will not discuss how the theory in this section can be adapted to infinite-dimensional settings, but see \cite[Remark 1.3.13, Remark 1.3.14]{mythesis}.}
\end{example}

\begin{remark}
Motivated by the last example, $\dagger$-categories were additionally introduced in the literature as $\C$-anti-linear involutions on $\C$-linear categories under the name `$*$-categories' \cite{ghez1985}. 
    This excludes bordism dagger categories, which are relevant examples in this work.
\end{remark}

A $\dagger$-functor between $\dagger$-categories is a functor $F$ which strictly commutes with the functor $\dagger$, i.e. $F(f^\dagger) = F(f)^\dagger$ for all morphisms.
A natural transformation is called \emph{isometric} if it evaluates to an isometry on every object.
This makes $\dagger$-categories into a $2$-category $\dagger\Cat$ in which equivalence is called \emph{$\dagger$-equivalence}.

There is also an analogous `weak' version of a dagger category, in which the equation $\dagger \circ \dagger = \id$ is replaced by a natural transformation.
We will refer to such a `weak' dagger category as an anti-involutive category.
This notion will be convenient because it behaves well under equivalences of categories.
More precisely, an anti-involution on a category is a fixed point for the $\Z/2$-action $\mathcal{C} \mapsto \mathcal{C}^{\op}$ on the $2$-category of categories.
Explicitly writing out this definition using the notion of fixed points in bicategories \cite[Appendix A]{mullerstehouwer} \cite{hesse2016frobenius}, results in the following.

\begin{definition}
\label{def:antiinvcat}
An \emph{anti-involution} on a category $\mathcal{C}$ consists of a functor $d\colon \mathcal{C} \to \mathcal{C}^{\op}$ together with a natural isomorphism $\eta\colon \id_{\mathcal{C}} \Rightarrow d^2$ such that $d\eta_c$ is the inverse of $\eta_{dc}$.
We call a category with anti-involution an \emph{anti-involutive category.}
\end{definition}

\begin{example}
\label{ex:vectorspaceantiinv}
    Let $\mathcal{C}$ be a symmetric monoidal category with duals.
    Then a choice of dual functor $x \mapsto x^*$ can be made into an anti-involution \cite[Lemma A.1.5]{mythesis}.
    The natural isomorphism $\eta$ witnesses the fact that $x$ and $x^{**}$ are both duals of $x^*$.
\end{example}

\begin{example}
    If $V$ is a complex vector space, we denote the \emph{complex conjugate vector space} by $\overline{V} = \{\bar{v}: v \in V\}$.
This vector space is equal to $(V,+)$ as an abelian group, but with complex conjugated scalar multiplication $z \bar{v} := \overline{\bar{z} v}$.
The operation $V \mapsto \overline{V}$ defines a functor on the category $\mathcal{C} = \Vect$ of finite-dimensional vector spaces if we set $\overline{T}(\overline{v}) := \overline{T(v)}$ for a linear map $T\colon V \to W$.
    The functor $V \mapsto \ol{V}^*$ is an anti-involution. The natural isomorphism $\eta_V: V \to \ol{\ol{V}^*}^*$ uses the isomorphism $V \to V^{**}$ given by evaluating functionals, the canonical isomorphism $\overline{\overline{V}} \cong V$ and the isomorphism $T\colon \ol{V^*} \cong \ol{V}^*$ given by
    \[
    T(\ol{f})(\ol{v}) := \ol{f(v)}.
    \]
\end{example}

\begin{remark}
    Calling $(\mathcal{C},d,\eta)$ an anti-involutive category is not standard terminology in the literature.
    In the surgery theory literature, this structure is sometimes called a category with duality.
    However, we find this terminology confusing: in many relevant examples, $dc$ will not be the dual of $c$ in any standard sense.
    Our terminology is motivated by what is sometimes called an anti-involution on an algebra $A$: a linear map $d\colon A \to A$ such that $d^2 = \id_A$ and $d(ab) = db \cdot da$.
\end{remark}

While dagger functors strictly preserve $\dagger$, anti-involutive functors are equipped with a datum specifying how the anti-involution is preserved:

\begin{definition}
\label{def:anti-invfunctor}
An \emph{anti-involutive functor} between anti-involutive categories $(\mathcal{C}_1, d_1, \eta_1),(\mathcal{C}_2, d_2, \eta_2)$ consists of a functor $F\colon \mathcal{C}_1 \to \mathcal{C}_2$ and a natural isomorphism $\phi\colon F \circ d_1 \Rightarrow d_2 \circ F$ such that the following diagram commutes:
    \begin{equation}
    \label{di:invfunctordata}
        \begin{tikzcd}[column sep = 60]
            F(x) \ar[d, "(\eta_2)_{F(x)}"] \ar[r, "F((\eta_1)_x)"] & (F \circ d_1 \circ d_1)(x) \ar[d, "\phi_{d_1(x)}"]\\
            (d_{2} \circ d_2 \circ F)(x) \ar[r, "d_2 (\phi_x)"] & (d_2 \circ F \circ d_1)(x) 
        \end{tikzcd}.
    \end{equation}
    An \emph{anti-involutive equivalence} is an anti-involutive functor which is also an equivalence of categories.
    \end{definition}

By \cite[Lemma 2.5]{jandagger}, an anti-involutive functor is an an anti-involutive equivalence if and only if it admits an anti-involutive inverse up to anti-involutive natural isomorphism.

\begin{example} 
\label{ex:oppositeantiinv}
    Let $(\mathcal{C}, d, \eta)$ be an anti-involutive category.
    Then $\mathcal{C}^{\op}$ is an anti-involutive category with inverted $\eta$.
    This makes $d$ into an anti-involutive functor $\mathcal{C} \to \mathcal{C}^{\op}$.
\end{example}

\begin{definition}
    An \emph{anti-involutive natural transformation} $u$ between anti-involutive functors $(F,\phi)$ and $(G,\psi)$ is a natural transformation such that the diagram
    \[
    \begin{tikzcd}
        (F \circ d_1)(x) \ar[d, "\phi"] \ar[r,"u_{d_1 x}"] & (G \circ d_1)(x) \ar[d, "\psi"]
        \\
        (d_2 \circ F)(x)  & (d_2 \circ G)(x) \ar[l,"d_2 u_{ x}"]
    \end{tikzcd}
    \]
    commutes for all objects $x$.
\end{definition}

Let $\aICat$ denote the $2$-category of anti-involutive categories.
There is a canonical $2$-functor 
\[
\dagger \Cat \to \aICat,
\]
which assigns the trivial coherence data $\eta_x = \id_x$.
This functor is not an equivalence, but it turns out it does admit a $2$-right adjoint 
\[
\Herm\colon \aICat \to \dagger \Cat,
\]
see Theorem \ref{th:herm2functor}.
If $(\mathcal{C},d,\eta)$ is anti-involutive, $\Herm (\mathcal{C},d,\eta)$ has a concrete description using what we call Hermitian pairings, the generalization of Hermitian pairings on vector spaces to arbitrary anti-involutive categories:

\begin{example}
\label{ex:vectherm}
    Let $\mathcal{C} = \Vect$ be the category of finite-dimensional vector spaces equipped with the anti-involution $\overline{(.)}^*$ of Example \ref{ex:vectorspaceantiinv}.
    Then an isomorphism $h\colon V \to \overline{V}^*$ is equivalent to a nondegenerate sesquilinear pairing $\langle \cdot ,\cdot \rangle$ on $V$.
    It is straightforward to verify that
    \[
\begin{tikzcd}
V \ar[r,"h"] \ar[d] & \ol{V}^*
\\
\ol{\ol{V}^*}^* \ar[ru,"\ol{h}^*"']& 
\end{tikzcd}
\]
commutes if and only if the equation
    \[
    \langle v,w \rangle = \overline{\langle w,v \rangle}
    \]
    holds for all $v,w \in V$.
\end{example}

This example motivates the following definition of a Hermitian pairing in an anti-involutive category, which we learned from \cite[Definition B.14.]{freedhopkins}.

\begin{definition}
\label{def:hermstr}
Let $(\mathcal{C},d, \eta)$ be an anti-involutive category.
A \emph{Hermitian pairing} in $(\mathcal{C},d, \eta)$ is defined to be an isomorphism $h\colon c \to dc$ such that
\begin{equation}
\begin{tikzcd}
c \ar[r,"h"] \ar[d,"\eta_c"] & dc
\\
d^2 c \ar[ru,"dh"']& 
\end{tikzcd}
\end{equation}
commutes.
\end{definition}

We assemble objects equipped with a Hermitian pairing into a category:

\begin{definition}
\label{def:Hermitian completion}
    The \emph{Hermitian completion} $\Herm \mathcal{C}$ of an anti-involutive category $\mathcal{C}$ is the category in which objects consists of Hermitian pairings $h\colon c \to dc$ and morphisms $(c_1, h_1) \to (c_2, h_2)$ are simply morphisms $c_1 \to c_2$ in $\mathcal{C}$.
\end{definition}

The Hermitian completion becomes a dagger category if we define the dagger of a morphism $f\colon (c_1, h_1) \to (c_2, h_2)$ as the composition
\begin{equation}
\label{eq:adjoint}
f^\dagger\colon c_2 \xrightarrow{h_2} dc_2 \xrightarrow{d f} dc_1 \xrightarrow{h_1^{-1}} c_1,
\end{equation}
see \cite[Lemma 3.4]{jandagger}.

\begin{example}
If $\mathcal{C} = \Vect$ and $d = \ol{(.)}^*$, then it follows by Example \ref{ex:vectherm} that the Hermitian completion is the category of finite-dimensional Hermitian vector spaces $\Herm_\C$.
It is straightforward to verify that also the dagger structure corresponds to the usual adjoint of linear maps.
In particular, the notions of isometric, unitary and self-adjoint morphisms in this dagger category give the usual notions of isometric, unitary and self-adjoint linear maps.
\end{example}

\begin{remark}
The fact that we required a morphism $(c_1, h_1) \to (c_2, h_2)$ in $\Herm(\mathcal{C},d,\eta)$ to simply be a morphism $c_1 \to c_2$ in $\mathcal{C}$ can be counter-intuitive.
This trivially implies that the canonical functor $\Herm (\mathcal{C},d,\eta) \to \mathcal{C}$ is fully faithful.
So when every object in $\mathcal{C}$ admits a Hermitian pairing, it is a confusing fact that it is an equivalence of categories.\footnote{This is even an equivalence of anti-involutive categories \cite[Lemma 4.6.]{jandagger}.}
For example, note that the functor $\Herm_\C \to \Vect$ is an equivalence of categories.
The reader might have expected us to require the obvious compatibility condition with the $h_i$ given by the diagram
\[
\begin{tikzcd}
    c_1 \ar[d,"h_1"] \ar[r,"f"] & c_2 \ar[d,"h_2"]
    \\
    dc_1 & dc_2 \ar[l,"df"]
\end{tikzcd}.
\]
This diagram is saying that $f$ is an isometry in the dagger category $\Herm (\mathcal{C},d,\eta)$.
For example, here $f$ is invertible if and only if $f$ is unitary.
The main point is thus that as a dagger category $\Herm (\mathcal{C},d,\eta)$ remembers much more than just this category $\mathcal{C}$.
It in particular remembers the groupoid of homotopy fixed points $(\mathcal{C}^{\cong})^{\Z/2}$ for the $\Z/2$-action of $d$ on the core as the groupoid of unitary isomorphisms of $\Herm$.
\end{remark}

From now on we will often implicitly assume that every object of $(\mathcal{C},d,\eta)$ admits a Hermitian pairing. 

\begin{example}
Consider a dagger category $\mathcal{C}$ as an anti-involutive category with $d=\dagger$ and $\eta_c = \id_c$.
The Hermitian completion of $\mathcal{C}$ has objects consisting of pairs of objects $c$ of $\mathcal{C}$ together with a self-adjoint automorphism $h\colon c \to dc = c^\dagger = c$.
The new dagger $\ddagger$ on $\Herm \mathcal{C}$ on a morphisms $f\colon (c_1, h_1) \to (c_2, h_2)$ is defined as $f^\ddagger = h_1^{-1} \circ f^\dagger \circ h_2$.
For example, starting with the dagger category of finite-dimensional Hilbert spaces, new objects are triples $(V,(\cdot ,\cdot ),A)$ consisting of a Hilbert space $(V,(.,.))$ and a self-adjoint invertible linear operator $A$ on $V$.
We can identify such triples with the not-necessarily-positive Hermitian pairing $(\cdot ,A \cdot )$ to realize a $\dagger$-equivalence between $\Herm(\Hilb)$ and $\Herm_\C$.
\end{example}

It turns out that the Hermitian completion construction has very nice categorical properties.

\begin{theorem}\cite[Lemma 3.13,Theorem 4.9]{jandagger}
\label{th:herm2functor}
The Hermitian completion extends to a $2$-functor
\[
\Herm\colon \aICat \to \dagger \Cat,
\]
which is strictly right adjoint to the canonical functor $\dagger \Cat \to \aICat$ in the sense that the triangle identities hold on the nose.
\end{theorem}

It can aid the intuition to reformulate the above theorem as a universal property: the Hermitian completion is the `cofree dagger category generated by a category with anti-involution':

\begin{corollary}
\label{cor:cofree}
    Let $(\mathcal{C},d,\eta)$ be an anti-involutive category and $\mathcal{D}$ a dagger category.
    Then there is an isomorphism of categories
    \[
    \Fun_{\aICat}(\mathcal{D}, \mathcal{C}) \cong \Fun_\dagger (\mathcal{D}, \Herm \mathcal{C}).
    \]
\end{corollary}

We introduce the following operation of transferring Hermitian pairings along isomorphisms.

\begin{definition}\cite[Definition 5.1]{jandagger}
\label{def:transfer}
    Let $h\colon c \to dc$ be a Hermitian pairing and $g\colon c' \to c$ any isomorphism.
    Then $dg \circ h \circ g$ is a Hermitian pairing on $c'$, which we will call the \emph{transfer of $h$ along $g$}.
\end{definition}

Transferring a Hilbert space pairing on a vector space along an invertible linear map $f\colon V_1 \to V_2$, amounts to modifying the pairing with the positive operator $f^\dagger f$.

Hermitian completions are in some sense `maximal' dagger categories, but we are often more interested in `smaller' dagger categories which contain very few Hermitian pairings.
Especially small are `minimal' dagger categories like $\Hilb$, in which unitary isomorphism classes agree with usual isomorphism classes.
We introduce the following terminology.

\begin{definition}
    A dagger category $\mathcal{D}$ is called \emph{minimal} if the map $\pi^U_0(\mathcal{D}) \to \pi_0(\mathcal{D})$ is bijective.
    A dagger category is called \emph{maximal} if is unitarily equivalent to a Hermitian completion.
\end{definition}

Given an anti-involutive category, we can restrict the collection of `allowed' Hermitian pairings on the Hermitian completion, similarly to how we can restrict the pairings on Hermitian vector spaces to the positive definite ones:

\begin{definition}
\label{def:positivitystructure}
    A \emph{positivity structure} on a category $\mathcal{C}$ with anti-involution $(d,\eta)$ is a collection $P$ of objects of $\Herm \mathcal{C}$ that surjects onto the objects of $\mathcal{C}$ under the forgetful map.
    We will call elements $h \in P$ \emph{positive (Hermitian) pairings}.
Given a positivity structure $P$ we denote by $\mathcal{C}_P \subseteq \Herm \mathcal{C}$ the full dagger subcategory on those objects $(c,h) \in \Herm \mathcal{C}$ such that $h \in P$.
    A positivity structure is \emph{closed} if it is closed under transfer in the sense of Example \ref{def:transfer}: for every $(h\colon c \to dc) \in P$ and every isomorphism $g\colon c' \to c$ also $dg \circ h \circ g \in P$. 
    Two positivity structures are \emph{equivalent} if they have the same closure.
\end{definition}

\begin{example}
\label{ex:nonclosedhilb}
    Consider the Hermitian completion of the category (equivalent to $\Vect$) in which objects are $\C^n$ and morphisms are given by matrices.
    Consider the positivity structure given by only allowing the standard inner product $\langle \cdot , \cdot \rangle_{st}$, giving us the dagger category $\mathcal{D} \subseteq \Hilb$ of Hilbert spaces of the form $(\C^n, \langle \cdot , \cdot \rangle_{st})$.
    The inclusion is an equivalence of dagger categories.
    However, this positivity structure is not closed.
    Its closure adds all inner products of the form $\langle . ,A^\dagger A . \rangle_{st}$ for $A$ an invertible matrix, which are in fact all positive definite inner products on $\C^n$.
\end{example}

\begin{remark}
    In Definition \ref{def:positivitystructure} of a positivity structure, we could have required the collection $P$ to only essentially surject onto the objects of $\mathcal{C}$.
     However, it is convenient in practice when every object of $\mathcal{C}$ admits at least one positive pairing.
     Moreover, this distinction is irrelevant for equivalence classes of positivity structures.
\end{remark}

\begin{example}
\label{ex:span}
Let $\mathcal{C}$ be a symmetric monoidal category with duals and let $d = (.)^*\colon \mathcal{C} \to \mathcal{C}^{\op}$ be the induced anti-involution.
    Then a Hermitian pairing on an object $c$ is the same as a self-duality $\ev_c\colon c \otimes c \to 1$, which is symmetric in the sense that it stays invariant under the braiding $\sigma_{c,c}$.
As a subexample, let $\mathcal{C}'$ be a category with pullbacks and a terminal object $\pt$.
Let $\mathcal{C} := \mathbf{Span} \, \mathcal{C}'$ be the category of spans.
    This category has objects $\obj \mathcal{C}$ and morphisms from $x$ to $y$ are equivalence classes of spans.
    Here a span is defined to be a pair of morphisms $(f,g)$ of the shape
    \[
    x \xleftarrow{f} z  \xrightarrow{g} y
    \]
    and an equivalence with a span $(z',f',g')$ is an isomorphism $z \cong z'$ making the obvious diagram commute.
    Composition of spans is given by pullback. 
    This category is symmetric monoidal under the cartesian product with monoidal unit $\pt$.
    Every object $x$ is canonically symmetrically self-dual.
    Indeed for the evaluation we can take $z = x$ with $f$ the diagonal map and $g$ the unique map to $\pt$, and similar for the coevaluation.
    The dagger category obtained by taking these self-dualities as the positivity structure, is the dagger category of spans with
    \[
    (x \xleftarrow{f} z  \xrightarrow{g} y)^\dagger = y \xleftarrow{g} z \xrightarrow{f} x.
    \]
\end{example}

Define the Hermitian isomorphism classes of an anti-involutive category $\mathcal{C}$ as $\pi_0^h(\mathcal{C}) := \pi_0^U( \Herm \mathcal{C})$.
By taking positivity structures, we can construct all possible dagger categories with a fixed underlying anti-involution up to unitary equivalence:

\begin{lemma}\cite[Corollary 5.7, Theorem 5.14]{jandagger}
\label{lem:compareposstrs}
Let $(\mathcal{C},d,\eta)$ be an anti-involutive category.
    There is a bijection between
    \begin{enumerate}
        \item equivalence classes $[P]$ of positivity structures;
        \item dagger categories $\mathcal{C}_P$ equipped with an anti-involutive equivalence $(\mathcal{C}_P,\dagger) \cong (\mathcal{C},d)$ modulo the following equivalence relation. 
        We say $\mathcal{C}_P$ and $\mathcal{C}_{P'}$ are equivalent when there exists a $\dagger$-equivalence $\mathcal{C}_{P} \cong \mathcal{C}_{P'}$ such that the triangle of anti-involutive functors
        \[
        \begin{tikzcd}
            \mathcal{C}_P \ar[rr] \ar[rd]& & \mathcal{C}_{P'} \ar[ld]
            \\
            & \mathcal{C} &
        \end{tikzcd}
        \]
        can be filled by an anti-involutive natural isomorphism;
        \item subsets $\pi_0^U(\mathcal{C}_P) \subseteq \pi_0^h(\mathcal{C})$ such that the composition
        \[
        \pi_0^U(\mathcal{C}_P) \subseteq \pi_0^h(\mathcal{C}) \to \pi_0(\mathcal{C})
        \]
        is surjective.
    \end{enumerate}
\end{lemma}

With the above lemma in mind, we will from now on assume all positivity structures are closed.

There are typically many choices for $P$ that make $\mathcal{C}_P$ minimal: 

\begin{example}
\label{ex:hermstrsonsvec}
    Let $\mathcal{C} = \Vect$ come equipped with the anti-involution $V \mapsto \ol{V}^*$.
    For every $d \in \Z_{\geq 0}$, pick a pair of integers $(p,q)$ such that $p+q =d$.
    We could then call the $d$-dimensional Hermitian vector space $\C^d$ with signature $(p,q)$ positive.
    This will result in a minimal dagger category.
Note that at this stage, there is no condition forcing signatures of vector spaces of different dimensions to be compatible.
Also note that some of these dagger categories are $\dagger$-equivalent, while some are not.
For example, the dagger category of finite-dimensional Hilbert spaces is $\dagger$-equivalent to the dagger category of finite-dimensional negative definite Hermitian vector spaces $\Hilb_{neg}$.
Note that this does not contradict Lemma \ref{lem:compareposstrs}.
Indeed, the two anti-involutive functors $(F,\phi),(F',\phi')\colon \Hilb \to \Vect$ in the diagram
\[
        \begin{tikzcd}
            \Hilb \ar[rr] \ar[rd]& & \Hilb_{neg} \ar[ld]
            \\
            & \Vect &
        \end{tikzcd}
        \]
satisfy $F = F'$, but $\phi = - \phi'$.
However, all natural automorphisms of $F$ are given by $u_\lambda (V,\langle.,. \rangle) := \lambda \id_V$ for some $\lambda \in \C^\times$.
The condition that $u_{\lambda}\colon (F,\phi) \Rightarrow (F',\phi')$ is anti-involutive is equivalent to $\lambda \ol{\lambda} = -1$, which is impossible.
\end{example}

\begin{definition}
\label{def:preservepositivity}
Let $(F,\phi)\colon (\mathcal{C}_1,d, P_1) \to (\mathcal{C}_2, d, P_2)$ be an anti-involutive functor between anti-involutive categories equipped with positivity structures.\footnote{From now on, we will often denote anti-involutions on different categories by the same letter.}
Then $F$ is said to \emph{preserve the positivity structures} if for all $(h\colon c \to dc) \in P_1$ the composition
\begin{equation}
\label{eq:F(P)}
F(c) \xrightarrow{F(h)} F(dc) \xrightarrow{\phi_c} dF(c)
\end{equation}
is in $P_2$.
\end{definition}

In general, anti-involutive functors always send a Hermitian pairing to the Hermitian pairing \eqref{eq:F(P)}. 
So $(F,\phi)\colon (\mathcal{C}_1,d) \to (\mathcal{C}_2, d)$ will map a subset $P$ of $\pi_0^h(\mathcal{C}_1)$ to some subset of $\pi_0^h(\mathcal{C}_2)$, which we will sloppily denote $F(P)$.
In particular, $F$ preserves positivity structures if $P_1$ will be mapped to a subset of $P_2$ under this procedure. 

We will now introduce terminology for endomorphisms in dagger categories that behave similarly to positive definite matrices.

\begin{definition}
\label{def:aut-positive}
    An endomorphism $f\colon c \to c$ is called 
    \begin{enumerate}
        \item \emph{positive} if it is of the form $g^\dagger g$ for some morphism $g\colon c \to c'$; 
        \item \emph{end-positive} if it is of the form $g^\dagger g$ for some endomorphism $g\colon c \to c$;
        \item \emph{iso-positive} if it is of the form $g^\dagger g$ for some isomorphism $g\colon c \to c'$;
        \item \emph{aut-positive} if it is of the form $g^\dagger g$ for some automorphism $g\colon c \to c$.
    \end{enumerate}
\end{definition}

In the above definition, only the first terminology is standard \cite{selinger2007positive}.
Note that every positive morphism is self-adjoint.

\begin{remark}
    By replacing $g$ with $g^\dagger$, we see that any of the above notions of positivity could have equivalently been stated with the dagger on the right morphism.
    Also note that being iso- and aut-positive is closed under taking inverses.
    However, being positive is not closed under composition in general. 
\end{remark}

We now provide some explicit equivalent conditions for the dagger category $\mathcal{C}_P$ to be minimal or maximal, motivated by the following example.

\begin{example}
\label{ex:positivesinHilb}
    Let $\mathcal{D} = \Hilb$ be the dagger category of finite-dimensional Hilbert spaces.
    A positive morphism in the dagger category of Hilbert spaces is a positive semidefinite operator, which in particular need not be invertible.
    Every positive morphism is end-positive.
    An automorphism of a Hilbert space $\mathcal{H}$ is iso-positive if and only if it is aut-positive if and only if it is a positive definite operator.
    In the dagger category $\mathcal{D} = \Herm$ of finite-dimensional Hermitian vector spaces, not all aut-positives are positive definite operators. In fact, all self-adjoint automorphisms are iso-positive.
\end{example}

Example \ref{ex:positivesinHilb} generalizes to the following statement.

\begin{lemma}
\label{lem:minimalcondition}
    \begin{enumerate}
        \item A dagger category is maximal if and only if every self-adjoint automorphism is iso-positive;
        \item A dagger category is minimal if and only if every iso-positive is aut-positive.
    \end{enumerate}
\end{lemma}
\begin{proof}
    The first point is shown in \cite[Section 4]{jandagger}.
    We prove the second point here.
    Let $\mathcal{D}$ be a dagger category.
    Since the map $\pi_0^U (\mathcal{D}) \to \pi_0(\mathcal{D})$ is surjective, $\mathcal{D}$ is minimal if and only if for every isomorphism $g\colon c \to c'$, there is a unitary isomorphism $\mu \colon c \to c'$.
    On the other hand, every iso-positive is aut-positive if and only if for every isomorphism $g\colon c \to c'$, there exists an automorphism $f\colon c \to c$ such that $f^\dagger f = g^\dagger g$.
    This expression can be rewritten as $(gf^{-1})^\dagger gf^{-1} = \id_c$.
    By setting $\mu = gf^{-1}$, we see that this is equivalent to the existence of the unitary isomorphism $\mu\colon c \to c'$.
\end{proof}

In other words, the notion of iso-positive in a dagger category can vary all the way between self-adjoint automorphism and aut-positive, depending on how large its positivity structure is.
Clearly we have the following implications
\[
\begin{tikzcd}
    \text{aut-positive} \ar[d] \ar[r] & \text{iso-positive} \ar[d] &
    \\
    \text{end-positive} \ar[r] & \text{positive} \ar[r] & \text{self-adjoint}
\end{tikzcd}
\]
Confusingly, it is not true in general that every end-positive isomorphism is aut-positive.
In a $\C$-linear setting this issue usually does not occur.
In particular, we have the following result.

\begin{proposition}
\label{prop:noninvminimal}
    Let $\mathcal{D}$ be a dagger category equipped with a conservative dagger functor $F\colon \mathcal{D} \to \Hilb$.
    Suppose that every positive endomorphism $f\colon c \to c$ is end-positive.
    Then $\mathcal{D}$ is minimal.
\end{proposition}
\begin{proof}
By Lemma \ref{lem:minimalcondition}, it suffices to show that every iso-positive automorphism $f\colon c \to c$ is aut-positive.
    Let $f\colon c \to c$ be an iso-positive automorphism, which is then clearly also positive.
    By assumption, there exists an endomorphism $g\colon c \to c$ such that $g^\dagger g = f$.
    If $g$ would not be an isomorphism, then $F(f)$ would not be an isomorphism because
    \[
    \det F(f) = \det F(g^\dagger g) = \det F(g)^\dagger \det F(g) = 0.
    \]
    Using the assumption that $F$ is conservative, we see that $f$ is aut-positive.
\end{proof}

\begin{remark}
    The minimality condition that every positive morphism is end-positive, is one axiom required for $C^*$-categories.
\end{remark}

\begin{remark}
    The converse of the above proposition is false: as explained in Example \ref{ex:hermstrsonsvec}, there are many positivity structures on $\Vect$ resulting in a minimal dagger category.
    However, most of these do not satisfy that every positive endomorphism $f\colon V \to V$ is end-positive.
    Indeed, suppose $P$ is a positivity structure, for which there exist $V_1$ and $V_2$ which are not both positive definite or both negative definite.
    Let $v_1 \in V_1$ and $v_2 \in V_2$ be vectors in some orthonormal basis of opposite norm.
    Let $T\colon V_1 \to V_2$ be the linear map sending $v_1$ to $v_2$ and all vectors orthogonal to $v_1$ to zero.
    Then $T^\dagger T$ is the negative of the orthogonal projection onto the line spanned by $v_1$, which is not positive as an operator.
    Therefore there are only two positivity structures on $\Vect$ so that every positive morphism is end-positive: the positive definite inner products and the negative definite inner products.
\end{remark}

The following example explains why we call $P$ a positivity structure:

\begin{example}
\label{ex:notionofposofdaggercat}
    If $\mathcal{D}$ is a $\dagger$-category considered as an anti-involutive category, recall that a Hermitian pairing is simply given by self adjoint automorphisms $h\colon c\to c^\dagger = c$.
    This anti-involutive category has a canonical closed positivity structure given by 
    \[
    \Pos(\mathcal{D}) := \{h\colon c \xrightarrow{\sim} c : h \text{ is iso-positive}\}.
    \]
    This positivity structure reproduces $\mathcal{D}$ inside its Hermitian completion.
    More precisely, the $\dagger$-functor including $\mathcal{D}$ into its Hermitian completion induces an $\dagger$-equivalence $\mathcal{D} \cong \mathcal{D}_{\Pos(\mathcal{D})}$.
    We can view $\mathcal{D}_{\Pos(\mathcal{D})}$ as a `closure' of $\mathcal{D}$ in the sense of Definition \ref{def:positivitystructure}, compare Example \ref{ex:nonclosedhilb}.
\end{example}

In the above example, we saw that any $\dagger$-category can up to $\dagger$-equivalence be presented in the form $\mathcal{C}_P$, where $P$ is a positivity structure on an anti-involutive category $\mathcal{C}$.
More generally, let $\aICat^{pos}$ be the $2$-category of anti-involutive categories with positivity structures in which anti-involutive functors are required to preserve the positivity structures.
There are $2$-functors 
\[
\begin{tikzcd}[column sep=40]
    \aICat^{pos} \arrow[r, bend left, "{(\mathcal{C},P) \mapsto \mathcal{C}_P}"] & \dagger \Cat \arrow[l,bend left,"{\mathcal{D} \mapsto (\mathcal{D}, \Pos(\mathcal{D}))}"].
\end{tikzcd}
\]

\begin{theorem}\cite[Theorem 5.14]{jandagger}
\label{th:jantheorem}
The above $2$-functors are inverse equivalences.
\end{theorem}

We conclude that much of the $2$-categorical theory of dagger categories can be captured by the theory of anti-involutions with positivity structures.
In the rest of this article, we tactfully and often implicitly switch between these two perspectives.

\subsection{Symmetric monoidal dagger categories}
\label{sec:symmon}

In this section, we provide results in the symmetric monoidal setting that are straightforward analogies to what we discussed in Section \ref{sec:Jan}, referring to Appendix \ref{app:symmondagger} for proofs.
Similarly to before, we define the $2$-category of symmetric monoidal anti-involutive categories to be the fixed points under the $\Z/2$-action $\mathcal{C} \mapsto \mathcal{C}^{\cop}$ on symmetric monoidal categories:

\begin{definition}
Define the $2$-category $\aICat_{\mathbb{E}_\infty}$ with 
\begin{itemize}
    \item objects: \emph{symmetric monoidal anti-involutive categories}, i.e. symmetric monoidal categories $\mathcal{C}$ equipped with a symmetric monoidal\footnote{In this article, `monoidal functor' will mean `strong monoidal functor'.} functor $d\colon \mathcal{C} 
\to \mathcal{C}^{\op}$ and a monoidal natural isomorphism $\eta: \id_{\mathcal{C}} \Rightarrow d^2$ such that $\eta_{dx} = d\eta_x$ for all $x \in \mathcal{C}$;
    \item $1$-morphisms: \emph{symmetric monoidal anti-involutive functors}, i.e. anti-involutive functors 
    \[
    (F,\phi)\colon (\mathcal{C}_1, d_1, \eta_1) \to (\mathcal{C}_2, d_2, \eta_2),
    \]
    which are also symmetric monoidal functors, such that the natural isomorphism $\phi\colon F \circ d \cong d \circ F$ of functors $\mathcal{C}_1 \to \mathcal{C}_2^{\op}$ is monoidal;
    \item $2$-morphisms: \emph{symmetric monoidal anti-involutive natural transformations}, i.e. natural transformations that are both anti-involutive and monoidal.
\end{itemize}
\end{definition}

Unpacking the above definition, we see that a symmetric monoidal anti-involution $d$ comes equipped with additional data
\[
\chi_{c_1, c_2}\colon d c_1 \otimes d c_2 \to d(c_1 \otimes c_2) \quad u\colon 1 \to d(1),
\]
satisfying various conditions.
Being a symmetric monoidal anti-involutive functor means that the diagrams
\begin{equation}\label{eq:monoidalunital}
\begin{tikzcd}[column sep = 33]
F(1_{\mathcal{C}_1}) \arrow[r,"{F(u_{\mathcal{C}_1})}"] & F(d1_{\mathcal{C}_1}) \arrow[r,"\phi_1"] & dF(1_{\mathcal{C}_1}) \arrow[d,"{d \epsilon}"]
\\
1_{\mathcal{C}_2} \arrow[u,"\epsilon"] \arrow[rr,"{u_{\mathcal{C}_2}}"] && d1_{\mathcal{C}_2}
\end{tikzcd}
\end{equation}
and
\begin{equation}
\label{eq:monoidalantiinvfunctor}
\begin{tikzcd}[column sep = 40]
F(dx) \otimes F(dx') \ar[r,"{\phi_x \otimes \phi_{x'}}"] \ar[d] & dF(x) \otimes dF(x') \ar[r,"{\chi_{F(x) \otimes F(x')}}"] & d(F(x) \otimes F(x'))
\\
F(dx \otimes dx') \arrow[r,"{F(\chi_{x,x'})}"] & F(d(x \otimes x')) \arrow[r,"{\phi_{x\otimes x'}}"] & dF(x\otimes x') \arrow[u]
\end{tikzcd}
\end{equation}
commute for all objects $x,x'$ of $\mathcal{C}_1$.
Here $\epsilon\colon 1_{\mathcal{C}_1} \to F(1_{\mathcal{C}_2})$ is the data of $F$ preserving the monoidal unit and $u_{\mathcal{C}_i}\colon 1_{\mathcal{C}_i} \to d 1_{\mathcal{C}_i}$ the data of $d$ preserving the monoidal unit for $i = 1,2$.
From now on we will often suppress data such as $d,\eta,\chi,u, \phi$ and $\epsilon$ from the notation, so the reader can expect statements of the form `let $\mathcal{C}$ be a symmetric monoidal anti-involutive category' or `$F$ is an anti-involutive functor'.

\begin{definition}
\label{def:monoidal dagger}
Define the $2$-category $\Cat^\dagger_{\mathbb{E}_\infty}$ with
\begin{itemize}
    \item objects: \emph{symmetric monoidal dagger categories}, i.e. dagger categories that are also symmetric monoidal categories such that $\otimes$ is a $\dagger$-functor and the unitor, the associator and the braiding are unitary;
    \item $1$-morphisms: \emph{symmetric monoidal dagger functors}, i.e. symmetric monoidal functors which are dagger functors such that $F(c_1) \otimes F(c_2) \to F(c_1 \otimes c_2)$ and $\epsilon\colon 1_{\mathcal{C}_2} \to F(1_{\mathcal{C}_1})$ are unitary;
    \item $2$-morphisms: \emph{symmetric monoidal isometric transformations}, i.e. natural transformations that are both isometries and monoidal.
\end{itemize}
\end{definition}

Let $(\mathcal{C}, d, \eta, \chi)$ be a symmetric monoidal anti-involutive category.
We define the tensor product of Hermitian pairings $h_1\colon c_1 \to dc_1$ and $h_2\colon c_2 \to dc_2$ by
\begin{equation}
\label{def:tensorproductofhermstrs}
c_1 \otimes c_2 \xrightarrow{h_1 \otimes h_2} dc_1 \otimes dc_2 \xrightarrow{\chi_{c_1,c_2}} d(c_1 \otimes c_2).
\end{equation}
The other data needed to make $\Herm \mathcal{C}$ into a symmetric monoidal dagger category is given by the symmetric monoidal structure of $\mathcal{C}$.
To obtain symmetric monoidal dagger categories that are not $\dagger$-equivalent to Hermitian completions, we need to discuss positivity structures.
The only thing to keep in mind, is that a full subcategory of a symmetric monoidal category generated by a collection of objects is only symmetric monoidal if it is closed under tensor product:

\begin{definition}
\label{def:monoidalposstr}
    Let $\mathcal{C}$ be a symmetric monoidal anti-involutive category.
    A positivity structure $P \subseteq \obj \Herm \mathcal{C}$ is called \emph{monoidal} if it is closed under tensor product.
\end{definition}

Note that the fact that $\Herm \mathcal{C}$ is a symmetric monoidal dagger category implies that $\pi_0^U(\Herm \mathcal{C})$ is a commutative monoid under tensor product.
The following lemma is clear.

\begin{lemma}
    The following are equivalent for a closed positivity structure $P$:
    \begin{enumerate}
        \item $P$ is monoidal;
        \item $P$ is a submonoid of $\pi_0^U(\Herm \mathcal{C})$;
        \item $\mathcal{C}_P$ is a symmetric monoidal dagger category.
    \end{enumerate}
\end{lemma}

\begin{example}
\label{ex:monoidaldaggerpositivity}
    Let $\mathcal{D}$ be a symmetric monoidal dagger category. 
    Then the associated canonical positivity structure in $\Herm \mathcal{D}$ of Example \ref{ex:notionofposofdaggercat} is monoidal, since $\otimes\colon \mathcal{D} \times \mathcal{D} \to \mathcal{D}$ is a $\dagger$-functor.
\end{example}

\begin{example}
For finite-dimensional vector spaces, the monoid structure on $\pi_0 (\Vect) = \N$ coming from the monoidal structure on $\Vect$ is given by multiplication.
The monoid structure on $\pi_0^U(\Herm \Vect) \cong \N \times \N$ is given by the formula telling us how the signature of a tensor product looks:
\[
(p_1,q_1) (p_2,q_2) = (p_1 p_2 + q_1 q_2, p_1 q_2 + p_2 q_1).
\]
Even after requiring positivity structures $P \subseteq \pi_0^U(\Herm \Vect)$ to be monoidal, there are still many positivity structures giving neither $\Herm_\C$ nor $\Hilb$.
One possibility would be that the positive definite and negative definite line are both contained in $P$, but their direct sum given by the hyperbolic 2-dimensional space is not.

To avoid these more pathological examples, we could require more compatibility conditions.
For example, if we require $P$ to be additionally closed under direct sums, then only $\Hilb$ and $\Herm_\C$ remain as possibilities.
Alternatively, we could follow a $C^*$-category style definition to completely single out $\Hilb$ here.
More precisely, we could require the monoidal dagger category to satisfy the stronger minimality requirement that every positive operator is end-positive, compare Proposition \ref{prop:noninvminimal}.
See \cite{matthewdagger} for a different approach to characterize dagger categories of Hilbert spaces.

Note however that these considerations are very specific about the $\C$-linear situation.
For example, we are agnostic on which bordism dagger categories should be considered `positive' in the sense $\Hilb$ might be called a `positive' dagger category.
Namely, bordism categories do not have direct sums and asking every positive morphism to be end-positive is typically unreasonable.
\end{example}

\begin{example}
\label{ex:oppositeposstr}
If $\mathcal{C}$ is a symmetric monoidal anti-involutive category with monoidal positivity structure $P$, then there is a canonical induced monoidal positivity structure $P^{\op}$ on $\mathcal{C}^{\cop}$ with the anti-involution from Example \ref{ex:oppositeantiinv} given by the inverses of elements of $P$.
This construction is in agreement with the obvious definition of the opposite of a symmetric monoidal dagger category in the sense that $(\mathcal{C}_P)^{\op} = (\mathcal{C}^{\op})_{P^{\op}}$.
\end{example}

Analogously to Theorem \ref{th:jantheorem}, we prove in Appendix \ref{app:symmondagger}:

\begin{theorem}
\label{th:symmetricjantheorem}
The Hermitian completion extends to a $2$-functor 
\[
    \aICat_{\mathbb{E}_\infty} \to \Cat^\dagger_{\mathbb{E}_\infty}
    \]
    right adjoint to the canonical functor $\Cat^\dagger_{\mathbb{E}_\infty} \to \aICat_{\mathbb{E}_\infty}$.
This induces an equivalence between $\Cat^\dagger_{\mathbb{E}_\infty}$ and the $2$-category of symmetric monoidal anti-involutive categories equipped with monoidal positivity structures:
    \[
    \Cat^\dagger_{\mathbb{E}_\infty} \cong (\aICat_{\mathbb{E}_\infty})_P.
    \]
\end{theorem}

\subsection{Dagger duality}
\label{sec:daggerduality}

In this section, we will study duals in symmetric monoidal dagger categories using the perspective developed in the previous sections.
With this goal in mind, we will first review dual functors on symmetric monoidal categories $\mathcal{C}$ and then study them on anti-involutive categories.
Dual functors on symmetric monoidal anti-involutive categories come equipped with canonical equivariance data for the anti-involution. 
This induces a canonical dual functor on the Hermitian completion, which will be a symmetric monoidal dagger functor.

We start by giving a lightning review on duals in symmetric monoidal categories. 
Let $x \in \mathcal{C}$ be an object of a symmetric monoidal category.
A \emph{dual} of $x$ is an object $x^* \in \mathcal{C}$ equipped with evaluation and coevaluation morphisms
\[
\ev_x\colon x^* \otimes x \to 1 \quad \coev_x\colon 1 \to x \otimes x^* 
\]
such that the triangle identities hold.
Duals are unique in the sense that if $(x', \ev'_x,\coev_x')$ is another dual, then there is a unique isomorphism $x' \cong x^*$ compatible with the respective evaluations and coevaluations.
Explicitly it is given by\footnote{We will often suppress associators and unitors from the notation to improve readability.}
\begin{equation}
\label{eq:dualsuniqueformula}
x^* \xrightarrow{\id_{x^*} \otimes \coev_x'} x^* \otimes x \otimes x' \xrightarrow{\ev_x \otimes \id_{x'}} x'.
\end{equation}
In particular, if $x^{*}$ is a dual of $x$ we obtain that $x$ is also a dual of $x^{*}$ under
\[
x \otimes x^{*} \xrightarrow{\sigma_{x, x^{*}}} x^{*} \otimes x \xrightarrow{\ev_{x}} 1,
\]
using the braiding of $\mathcal{C}$.
If $x^{**}$ is a dual of $x^*$, we can use uniqueness of duals to get an isomorphism $x \cong x^{**}$, which by \eqref{eq:dualsuniqueformula} is given as
    \begin{equation}
\label{eq:canonicalpivotalstr}
    x \xrightarrow{\coev_{x^*} \otimes \id_x} x^* \otimes x^{**} \otimes x \xrightarrow{\id_{x^*} \otimes \sigma_{x^{**}, x}} x^* \otimes x \otimes x^{**} \xrightarrow{\ev_x \otimes \id_{x^{**}}} x^{**}.
    \end{equation}
If every object of $\mathcal{C}$ admits some dual, we say that $\mathcal{C}$ \emph{has duals}.
After picking a specific dual for every object of $\mathcal{C}$, we can construct a symmetric monoidal functor 
\[
(.)^*\colon \mathcal{C} \to \mathcal{C}^{\op}.
\]
The dual of a morphism $f\colon c_1 \to c_2$ is the unique morphism $f^*\colon c_2^* \to c_1^*$ making the diagram 
\[
\begin{tikzcd}[column sep = 35]
    c_2^* \otimes c_1 \ar[d,"f^* \otimes \id_{c_2}"'] \ar[r, "\id_{c_2^*} \otimes f"] & c_2^* \otimes c_2 \ar[d,"\ev_{c_2}"]
    \\
    c_1^* \otimes c_1 \ar[r,"\ev_{c_1}"] & 1
\end{tikzcd}
\]
commute.
The monoidal data of the functor $(.)^*$ is constructed by uniqueness of duals, using the fact that $x_1^* \otimes x_2^*$ becomes a dual of $x_1 \otimes x_2$ after applying the braiding.
Note that a dual functor on $\mathcal{C}$ induces a dual functor on $\mathcal{C}^{\op}$ in which the evaluation and coevaluation maps are exchanged.

\begin{definition}
    Let $\mathcal{C}$ be symmetric monoidal category.
    Then a choice of dual for every object induces a symmetric monoidal functor
    $(.)^*\colon \mathcal{C} \to \mathcal{C}^{\op}$
    called a \emph{dual functor}.
\end{definition}

The statement of uniqueness of duals can now be strengthened to obtain that any two dual functors are equivalent:

\begin{proposition}
\label{prop:dualsuniqeu}
Let $(.)^*$ be a dual functor coming from the choices
\[
\ev\colon x^* \otimes x \to 1 \quad \coev\colon 1 \to x \otimes x^*
\]
for all $x \in \mathcal{C}$.
    Another choice of duals 
    \[
\ev'\colon x' \otimes x \to 1 \quad \coev'\colon 1 \to x \otimes x'
\]
for every object induces a unique monoidal natural isomorphism between the two induced dual functors $(.)' \cong (.)^*$ with the property that it intertwines the respective evaluation and coevaluation maps.
Conversely, if $F\colon \mathcal{C} \to \mathcal{C}^{\op}$ is any functor equipped with a monoidal natural isomorphism $F \cong (.)^*$, there is a unique way to make $F(x)$ into a dual of $x$ such that this natural isomorphism intertwines the respective evaluation and coevaluation maps.
\end{proposition}

In particular, we see that $F$ being a dual functor is equivalent to the condition that $F(x)$ can be made into a dual of $x$ compatibly with the isomorphisms $F(x \otimes y) \cong F(x) \otimes F(y)$ and $F(1) \cong 1$, such that $F(f)$ is the dual of $f$ using the given evaluation maps.
    Now, if $(.)^*\colon \mathcal{C}_1 \to \mathcal{C}_1^{\op}$ is a dual functor and $F\colon \mathcal{C}_1 \to \mathcal{C}_2$ is a symmetric monoidal functor, then $F(x)^*$ and $F(x^*)$ are both dual to $F(x)$.
    This allows us to construct a canonical monoidal natural isomorphism $F \circ (.)^* \Rightarrow (.)^* \circ F$ similarly to Proposition \ref{prop:dualsuniqeu}.
Note however that we used the braiding in making the dual functor monoidal, as well as in the isomorphism $x \cong x^{**}$.
Therefore with our conventions, these data are only preserved by \emph{symmetric} monoidal functors.

\begin{example}
    Let $\Vect_\C$ be the symmetric monoidal category of finite-dimensional vector spaces.
    Let $\ev_V$ and $\coev_V$ denote the standard evaluation and coevaluation maps giving us a symmetric monoidal dual functor $(.)^*$.
    Fix $\alpha \in \C^\times$ and define new evaluation and coevaluation maps by $\ev'_V := \alpha \circ \ev_V$ and $\coev_V := \coev_V \circ \alpha^{-1}$.
    It is straightforward to verify the zigzag equations.
    This gives us a second dual functor $(.)'$ which is equal to $(.)^*$ as a functor.
    However, note that the isomorphisms $1' \cong 1$ and $(x \otimes y)' \cong x' \otimes y'$ are modified by multiplication with $\alpha$.
    The natural isomorphism $\eta\colon (.)^* \Rightarrow (.)'$ specifying uniqueness of duals is given by multiplication with $\alpha$, which is indeed a monoidal natural transformation.
\end{example}

We now discuss duality in the context of symmetric monoidal dagger categories, adopting the terminology of \cite{penneysUDF}. 

\begin{definition}
    Let $\mathcal{D}$ be a symmetric monoidal dagger category with duals.
    A choice of dual functor $\mathcal{D} \to \mathcal{D}^{\cop}$ is called a \emph{unitary dual functor} if it is a symmetric monoidal dagger functor. 
\end{definition}

Note that it is a condition for a dual functor on a symmetric monoidal dagger category to be a unitary dual functor: we need that $f^{*\dagger} = f^{\dagger *}$ for all morphisms and that the preferred isomorphism $(x \otimes y)^* \cong x^* \otimes y^*$ is unitary for all $x,y \in \mathcal{D}$.
    As emphasized in \cite{penneysUDF}, unitary dual functors are subtle.
    For example, the canonical natural isomorphism between two choices of unitary dual functor need not be unitary.
    We argue that the perspective using anti-involutions sheds some light on these issues. 
Therefore, we study dual functors on anti-involutive categories and how they preserve Hermitian pairings.

\begin{lemma}
\label{lem:canonicalcandidateUDF}
Let $\mathcal{C}$ be a symmetric monoidal anti-involutive category which has duals.
A choice of dual functor $(.)^*\colon \mathcal{C} \to \mathcal{C}^{\cop}$ is canonically anti-involutive. 
\end{lemma}
\begin{proof}
Let $(.)^*\colon \mathcal{C} \to \mathcal{C}^{\cop}$ be a dual functor.
Recall that uniqueness of duals provides a canonical monoidal natural isomorphism $d \circ (.)^* \Rightarrow (.)^* \circ d$ of symmetric monoidal functors $\mathcal{C} \to \mathcal{C}$.
We then only have to show that the diagram
\[
\begin{tikzcd}
    x^* \ar[d,"\eta_{x^*}"] & (d^2 x)^* \ar[l,"\eta_x^*"]
    \\
    d^2(x^*)  & d(dx^*) \ar[u] \ar[l]
\end{tikzcd}
\]
commutes.
This follows from the fact that $\eta$ is a monoidal natural isomorphism, see the next lemma.
\end{proof}

\begin{lemma}
\label{lem:dualofnattr}
        Let $F, G\colon \mathcal{C} \to \mathcal{D}$ be monoidal functors with a monoidal natural isomorphism $\phi\colon F \Rightarrow G$.
    Then, the following diagram commutes
    \[
    \begin{tikzcd}
    F(x^*) \ar[r,"{\phi_{x^*}}"] \ar[d] & G(x^*) \ar[d]
    \\
    F(x)^* & \ar[l,"{\phi_x^*}"] G(x)^*
    \end{tikzcd}.
    \]
\end{lemma}
\begin{proof}
    It follows by a diagram chase that the composite isomorphism 
    \[
        F(x^*) \xrightarrow{\phi_{x^*}} G(x^*) \to G(x)^* \xrightarrow{\phi_x^*} F(x)^*
    \]
    satisfies the property that characterises $F(x^*) \cong F(x)^*$. 
\end{proof}

\begin{remark}
    Recall that a dual on $\mathcal{C}$ induces a dual on $\mathcal{C}^{\op}$ in which evaluation maps and coevaluation maps are exchanged.
Therefore $d$ preserving duals in Lemma \ref{lem:canonicalcandidateUDF} means it maps evaluation maps to coevaluation maps.
\end{remark}

\begin{remark}
\label{rem:udfsequivalent}
    Note that the anti-involutive dual functor of an anti-involutive symmetric monoidal category with duals is completely canonical in the following precise sense.
    Let $(.)^\vee\colon \mathcal{C} \to \mathcal{C}^{\cop}$ be another choice of dual functor and let $\sigma_x\colon x^* \cong x^\vee$ be the monoidal natural isomorphism as in Proposition \ref{prop:dualsuniqeu}.
    Then $\sigma_x$ is an anti-involutive natural isomorphism.
    Indeed, 
    \[
    \begin{tikzcd}
        d(c^*) \ar[r] & d(c)^* \ar[d,"\sigma_{dc}"]
        \\
        d(c^\vee) \ar[u,"d \sigma_c"] \ar[r] & d(c)^\vee
    \end{tikzcd}
    \]
    commutes because all natural isomorphisms involved are uniqueness of dual isomorphisms for duals of $dc$.
\end{remark}

\begin{remark}
\label{rem:lambdasubtle}
    We can apply Lemma \ref{lem:canonicalcandidateUDF} to the special case where the anti-involutive category is a symmetric monoidal dagger category $\mathcal{D}$.
    If we fix a dual functor, then formula \eqref{eq:dualsuniqueformula} tells us that the canonical natural isomorphism $\lambda\colon (.)^{* \dagger} \cong (.)^{\dagger *}$ is given by
    \begin{equation}
    \label{eq:canonicaldualautom}
   x^* \xrightarrow{\id_{x^*} \otimes \coev_x} x^* \otimes x \otimes x^* \xrightarrow{\sigma_{x^*,x} \otimes \id_{x^*}} x \otimes x^* \otimes x^* \xrightarrow{\coev^\dagger_x \otimes \id_{x^*}}x^*.
    \end{equation}
    Because $\lambda$ expresses uniqueness of duals, it is well-behaved categorically, but in a subtle way.
    For example, note how the fact that the isomorphism $\lambda\colon (.)^{* \dagger} \cong (.)^{\dagger *}$ is natural, seems to imply at first sight that the automorphism \eqref{eq:canonicaldualautom} measures the failure of morphisms $f\colon x_1 \to x_2$ satisfying $f^{\dagger *} = f^{* \dagger}$.
    In particular, if $\lambda$ is the identity on all objects, then $(.)^*$ is a unitary dual functor.
    This follows because natural isomorphisms specifying uniqueness of duals are monoidal.
    However, $\lambda$ need not be the identity even for a unitary dual functor.
\end{remark}

\begin{definition}
    \label{def:standardUDF}
    Let $\mathcal{C}$ be a symmetric monoidal anti-involutive category with duals.
    Then a \emph{standard unitary dual functor} on $\Herm \mathcal{C}$ is the symmetric monoidal dagger functor induced by the symmetric monoidal anti-involutive structure of Lemma \ref{lem:canonicalcandidateUDF} on a choice of dual functor $(.)^*$ on $\mathcal{C}$ under Theorem \ref{th:symmetricjantheorem}.
\end{definition}

\begin{remark}
    Because the morphisms of $\Herm \mathcal{C}$ are simply the underlying morphisms of $\mathcal{C}$, the standard unitary dual functor is clearly a dual functor.
\end{remark}

\begin{remark}
    By Remark \ref{rem:udfsequivalent} and the correspondence between anti-involutive and unitary natural isomorphisms, we know that given two standard unitary dual functors, there is a canonical unitary natural isomorphism between them.
\end{remark}

Next we want to generalize the discussion from the Hermitian completion to general dagger categories.
For this, we first make some observations on what the natural Hermitian pairing on the dual $c^*$ looks like explicitly.
First note that there is a natural Hermitian pairing on $dc$:

\begin{example}[{\cite[Example 3.10]{jandagger}}]
\label{ex:hermstondc}
Let $(\mathcal{C},d,\eta)$ be an anti-involutive category and $c$ an object of $\mathcal{C}$.
    If $h\colon c \to dc$ is a Hermitian pairing, on $c$, then $(dh)^{-1}\colon dc \to d^2 c$ is a Hermitian pairing on $dc$.
    Moreover, $h$ is a unitary isomorphism between $c$ and $dc$.
\end{example}

The Hermitian pairing on the dual is similar in spirit:

\begin{definition}
\label{def:dualhermstr}
    Let $\mathcal{C}$ be a symmetric monoidal anti-involutive category equipped with its canonical anti-involutive dual functor and let $h\colon c \to dc$ be a Hermitian pairing.
    The \emph{dual Hermitian pairing} on $c^*$ is given by the composition
    \[
c^* \xrightarrow{h^{*-1}} (dc)^* \cong d(c^*),
\]
where the isomorphism used the anti-involutivity data of Lemma \ref{lem:canonicalcandidateUDF}.
\end{definition}

Note that the standard unitary dual functor on the Hermitian completion maps an object $(c,h)$ to $c^*$ together with the dual Hermitian pairing of $c$ in the sense of Definition \ref{def:dualhermstr}.
The standard dual functor induces an anti-involution of monoids 
    \[
    \pi_0^h (\mathcal{C}) \to \pi_0^h (\mathcal{C}^{\op}) \cong \pi_0^h (\mathcal{C})^{\op} \quad P \mapsto P^*,
    \]
    compare with the notation explained under Definition \ref{def:preservepositivity}.
Using the isomorphism of monoids $\pi_0^h (\mathcal{C})^{\op} \cong \pi_0^h (\mathcal{C})$ given by taking inverses, we can view this also as an involution on $\pi_0^h (\mathcal{C})$, compare Example \ref{ex:oppositeposstr}.

    In general, there is no reason for $c^*$ to be isomorphic to $c$, so that $P \mapsto P^*$ can cover a nontrivial involution of monoids $\pi_0 \mathcal{C} \to \pi_0 \mathcal{C}$.
But even if $c^* \cong c$ and $h$ is a Hermitian pairing on $c$, then $c$ might not be unitarily isomorphic to $c^*$ with the dual Hermitian pairing.
We will see in Section \ref{sec:shilb} that this is the case for super vector spaces.
    However, fixed points $P = P^*$ for the involution on $\pi_0^h(\mathcal{C})$ do give standard unitary dual functors on $\mathcal{C}_P$.
    Indeed, note that the standard unitary dual functor restricts to a functor 
    \[
    \mathcal{C}_P \to \mathcal{C}^{\op}_{P^{\op}}
    \]
    if and only if $P^* \subseteq P$.
    Since the dual functor is an equivalence of anti-involutive categories, this happens if and only if $P^* = P$.
    This situation will be studied in the next section.

\begin{remark}
    Every unitary dual functor on $\mathcal{C}_P$ recovers the dual functor on $\mathcal{C}$, but that it need not recover its canonical anti-involutive data.
    In other words, there can be potentially interesting unitary dual functors on $\mathcal{C}_P$ that come from other anti-involutive dual functors.
    This in particular is the case for categories of fermionic nature such as $\sHilb$ and `spin-like' bordism categories, as we will see in Section \ref{sec:fermdagger}.
\end{remark}

\subsection{Dagger compactness}
\label{sec:daggercompact}

Next we relate the discussion of the last section with what are called dagger compact categories in the literature \cite{selinger2007positive}.

\begin{definition}
\label{def:moderndaggercpt}
    Let $\mathcal{D}$ be a symmetric monoidal dagger category with duals, which comes presented in the form $\mathcal{D} \cong \mathcal{C}_P$, for a symmetric monoidal anti-involutive category $\mathcal{C}$ and monoidal positivity structure $P$. 
    We say that $\mathcal{D}$ is \emph{dagger compact} if $P^* = P$.
\end{definition}

Note that $\mathcal{C}_P$ is dagger compact if and only if the standard unitary dual functor on $\Herm \mathcal{C}$ restricts to $\mathcal{C}_P$.
In other words, for all $h \in P$, the dual Hermitian pairing \ref{def:dualhermstr} is again in $P$.
To show the above definition is well-defined, we need a couple of lemmas that are proven below.

\begin{lemma}
\label{lem:daggercptwelldefd}
Dagger compactness is a well-defined property of the symmetric monoidal dagger category $\mathcal{D}$, which is preserved under symmetric monoidal $\dagger$-equivalences.
\end{lemma}
\begin{proof}
It follows by Theorem \ref{th:symmetricjantheorem} that every symmetric monoidal dagger category $\mathcal{D}$ is $\dagger$-equivalent to one of the form $\mathcal{C}_P$.
To show the lemma, it thus suffices to show that if $F\colon\mathcal{C} \cong \mathcal{C}'$ is an equivalence of anti-involutive categories mapping the monoidal positivity structure $P$ to the monoidal positivity structure $P'$ and $\mathcal{C}_P$ is dagger compact, then so is $\mathcal{C}'_{P'}$.
    For this consider the natural isomorphism $\zeta_x\colon F(x^*) \cong F(x)^*$ specifying uniqueness of duals filling the square
\[
 \begin{tikzcd}
        \mathcal{C} \ar[r,"F"] \ar[d,"(.)^*"] &  \mathcal{C}' \ar[d,"(.)^*"] \ar[dl,Rightarrow, "\zeta", shorten <=3, shorten >=3]
        \\
        \mathcal{C}^{\cop} \ar[r,"F"] &  (\mathcal{C}')^{\cop}
    \end{tikzcd}.
\]
By Lemma \ref{lem:lemma2}, this is a diagram in the $2$-category of symmetric monoidal anti-involutive categories.
Therefore, Lemma \ref{lem:lemma1} applies to the case where $G_1(x) = F(x)^*$ and $G_2(x) = F(x^*)$.
It follows that 
\[
P' = F(P) = F(P^*) = F(P)^* = (P')^*.
\]
\end{proof}

\begin{lemma}
\label{lem:lemma2}
    Let $\mathcal{C}_1, \mathcal{C}_2$ be symmetric monoidal anti-involutive categories with duals, which we equip with their respective anti-involutive dual functors given by Lemma \ref{lem:canonicalcandidateUDF}.
    Let $F\colon \mathcal{C}_1 \to \mathcal{C}_2$ be a symmetric monoidal anti-involutive functor.
    Denote by $\zeta_x\colon F(x^*) \cong F(x)^*$ the canonical isomorphism saying that $F$ preserves duals.
    Then $\zeta$ is a monoidal anti-involutive natural isomorphism.
\end{lemma}
\begin{proof}
    Let $\phi\colon F(dx) \to dF(x)$ denote the monoidal natural isomorphism making $F$ into an anti-involutive functor.
    The data $d(x^*) \cong d(x)^*$ making $(.)^*$ anti-involutive is the isomorphism specifying uniqueness of duals.
    To show $\zeta$ is anti-involutive, we have to show that the diagram of isomorphisms
    \[
    \begin{tikzcd}
        F((dx)^*) \ar[r,"\zeta_{dx}"] & F(dx)^* 
        \\
        F(d(x^*)) \ar[d,"\phi_{x^*}"] \ar[u] & (dF(x))^* \ar[u,"\phi^*_x"]
        \\
        dF(x^*) & d(F(x)^*) \ar[l,"d\zeta_x"] \ar[u]
    \end{tikzcd}
    \]
    commutes.
    This follows from applying Lemma \ref{lem:dualofnattr} to the monoidal natural isomorphism $\phi\colon F \circ d \Rightarrow d \circ F$.
    Indeed, note that the compositions 
    \[
    F(d(x^*)) \cong F((dx)^*) \cong F(dx)^* \quad  \text{  and  } \quad dF(x^*) \cong d(F(x)^*) \cong (dF(x))^*
    \]
    are exactly the isomorphisms that specify how $F \circ d$ respectively $d \circ F$ preserve duals.
    Since natural isomorphisms specifying uniqueness of duals are monoidal, this finishes the proof.
\end{proof}

\begin{lemma}
\label{lem:lemma1}
Let $\mathcal{C}_1, \mathcal{C}_2$ be symmetric monoidal anti-involutive categories and let $P_1$ be a monoidal positivity structure on $\mathcal{C}_1$.
    Let $\zeta\colon G_1 \Rightarrow G_2$ be a monoidal anti-involutive natural isomorphism between symmetric monoidal anti-involutive functors $G_1, G_2\colon \mathcal{C}_1 \to \mathcal{C}_2$.
    Then $G_1(P_1) = G_2(P_1)$ in the closure.
\end{lemma}
\begin{proof}
Let $(\zeta_1)_x\colon G_1(dx) \to d G_1(x)$ and $(\zeta_2)_x\colon G_2(dx) \to d G_2(x)$ be the data making $G_1$ and $G_2$ into anti-involutive functors.
Let $(h_1\colon x \to dx) \in P_1$ be a positive pairing.
    The diagram 
    \[
    \begin{tikzcd}
        G_1(x) \ar[r, "\zeta_x"] \ar[d,"G_1(h_1)"] & G_2(x) \ar[d,"G_2(h_1)"]
        \\
        G_1(dx) \ar[r,"\zeta_{dx}"] \ar[d,"(\zeta_1)_x"] & G_2(dx)  \ar[d,"(\zeta_2)_x"]
        \\
        d G_1(x) & dG_2(x) \ar[l,"d\zeta_{x}"]
    \end{tikzcd}
    \]
    commutes: the upper square by naturality of $\zeta$, the lower part by the fact that $\zeta$ is an anti-involutive natural transformation.
    The leftmost vertical composition is by definition the image of $h_1 \in P_1$ under $G_1$ and the rightmost vertical composition is the image of $h_1 \in P_1$ under $G_2$.
    Because $\zeta_x$ is an isomorphism we see that these two images are transfers of each other.
    We see that if $h_2 \in G_2(P)$, then a transfer of it is in $G_1(P)$ and vice-versa.
    By closure of positivity structures, we see that $G_1(P) = G_2(P)$.
\end{proof}

\begin{example}
    Every Hermitian completion is dagger compact.
    In particular, the dagger category $\Herm_\C$ of finite-dimensional Hermitian vector spaces is dagger compact.
\end{example}

\begin{example}
    The dagger category $\Hilb$ of finite-dimensional Hilbert spaces is dagger compact.
\end{example}

\begin{example}
    The oriented bordism category with their dagger given by reversing the orientation and the direction of bordisms is dagger compact.
\end{example}

\begin{example}
    The span category of Example \ref{ex:span} is dagger compact, also see \cite{baezspans}.
\end{example}

For connecting Definition \ref{def:moderndaggercpt} with more common approaches to dagger compactness in the literature \cite{selinger2007positive, cockett2021dagger}, the following results will be convenient.

\begin{lemma}
\label{lem:compactcondition}
Let $\mathcal{D}$ be a symmetric monoidal dagger category with duals and let $(.)^*\colon \mathcal{D} \to \mathcal{D}$ be a dual functor.
        Then for all objects $x \in \mathcal{D}$ and dualities $\ev_x,\coev_x$ on $x$ the diagram 
    \[
    \begin{tikzcd}[column sep = 50]
        1   & x^* \otimes x \ar[d,"\sigma_{x^*,x}"] \ar[l,"\ev_x"]
        \\
        x \otimes x^* \ar[u,"\coev_x^\dagger"] \ar[r, "\id_{x} \otimes \lambda_x"]  & x \otimes x^*
    \end{tikzcd}
    \]
    commutes.
    Here $\lambda$ is the automorphism expressing uniqueness of duals $(.)^{* \dagger} \cong (.)^{\dagger *}$, compare Remark \ref{rem:lambdasubtle}.
\end{lemma}
\begin{proof}
    By the explicit expression \eqref{eq:canonicaldualautom} for $\lambda$, we see that we have to show the diagram
    \[
    \begin{tikzcd}[column sep = 50]
        x \otimes x^* \ar[rr,"\coev_x^\dagger"] \ar[d,"\sigma_{x,x^*}", swap] & & 1
        \\
        x^* \otimes x \ar[d,"\id_{x^*} \otimes \coev_x \otimes \id_x", swap] &  x^* \otimes x \ar[lu,"\sigma_{x^*,x}"] & x^* \otimes x \ar[u,"\ev_x"]
        \\
        x^* \otimes x \otimes x^* \otimes x \ar[rr,"{\sigma_{x^*,x} \otimes \id_{x^* \otimes x}}", swap] \ar[ru,"\id_{x^* \otimes x} \otimes \ev_x", swap]  && x \otimes x^* \otimes x^* \otimes x \ar[u,"\coev^\dagger_x \otimes \id_{x^* \otimes x}",swap]
    \end{tikzcd}
    \]
    commutes.
    The left part commutes by the triangle identity and the fact that the braiding is symmetric. 
    The right part commutes by the interchange law.
\end{proof}

\begin{proposition}
\label{prop:daggercptTFAE}
The following are equivalent for a symmetric monoidal dagger category $\mathcal{D}$ with duals:
\begin{enumerate}[(1)]
    \item $\mathcal{D}$ is dagger compact;
    \item Let $\lambda$ be the isomorphism expressing uniqueness of duals corresponding to the standard dual functor on $\Herm \mathcal{D}$.
    Then for all $x \in \mathcal{D} \subseteq \Herm \mathcal{D}$, $\lambda_x$ is iso-positive;
    \item every object $x \in \mathcal{D}$ admits a dual $(\ev_x\colon x^* \otimes x \to 1,\coev_x\colon 1 \to x \otimes x^*)$ such that the diagram
    \begin{equation}
    \label{eq:classicaldaggercpt}
    \begin{tikzcd}[column sep = 50]
        1   & x^* \otimes x \ar[dl,"\sigma_{x^*,x}"] \ar[l,"\ev_x"]
        \\
        x \otimes x^* \ar[u,"\coev_x^\dagger"]   &
    \end{tikzcd}
    \end{equation}
    commutes;
\end{enumerate}
\end{proposition}
\begin{proof}
Recall that $\mathcal{D}$ is $\dagger$-equivalent to the symmetric monoidal dagger category $\mathcal{D}_{\Pos}$ explained in Example \ref{ex:notionofposofdaggercat}, which in our case is symmetric monoidal.
    We have fully faithful inclusions of symmetric monoidal dagger categories
    \[
    \mathcal{D} \hookrightarrow \mathcal{D}_{\Pos} \hookrightarrow \Herm \mathcal{D}.
    \]
    Let $\ev_x\colon x^* \otimes x \to 1$ denote arbitrary choices of duals.
    Explicitly writing out the definition of $P^*$ for the special case where $d$ is $\dagger$ gives us that $P^* \subseteq P$ if and only if for all $h\colon  x \to x$ iso-positive, the isomorphism 
    \[
    x^* \xrightarrow{h^{*-1}} x^* \xrightarrow{\lambda_x} x^*
    \]
    is iso-positive.
    The reader should be warned at this point that $\lambda_x$ is computed here with respect to the dagger structure on $\Herm \mathcal{D}$ and in particular depends on the chosen object $h\colon x \to dx$ of $\Herm \mathcal{D}$.
    However, since $\mathcal{D} \hookrightarrow \mathcal{D}_{\Pos}$ is unitarily essentially surjective, it suffices to check this condition only for $h = \id_x$.
    We obtain $(1) \iff (2)$.

    For $(2) \implies (3)$, first assume that $\lambda_x$ is iso-positive for all $x$.
    Let $x \in \mathcal{D}$ and let $f\colon  x^* \to x'$ be an isomorphism such that $\lambda_x^{-1} = f^\dagger f$. 
    Realize $x'$ as the dual of $x$ by
    \begin{align*}
    x' \otimes x \xrightarrow{f^{-1} \otimes \id_x} x^* \otimes x \xrightarrow{\ev_x} 1
    \\
    1 \xrightarrow{\coev_x} x \otimes x^* \xrightarrow{\id_x \otimes f} x \otimes x'.
    \end{align*}
    It is straightforward to check the triangle identities.
    Now note that the diagram
    \[
    \begin{tikzcd}[column sep =35]
        1 & x^* \otimes x \ar[l,"\ev_x"] \ar[d,"\sigma_{x^*,x}", swap] 
        \\
        x \otimes x^* \ar[u,"\coev_x^\dagger"] \ar[r,"\id_x \otimes \lambda_x"]  & x \otimes x^*  \ar[dl,"\id_x \otimes f"] 
        \\
        x \otimes x' \ar[u,"\id_x \otimes f^\dagger"] \ar[r,"\sigma_{x,x'}",swap] & x' \otimes x \ar[uu, bend right=60,"f^{-1} \otimes \id_x", swap]
    \end{tikzcd}
    \]
    commutes using Lemma \ref{lem:compactcondition} and naturality of the braiding.
    It follows that the duality between $x'$ and $x$ satisfies condition $(3)$.
    Conversely, suppose we have chosen dualities on all objects satisfying $(3)$.
    Replacing the $\coev_x^\dagger$ in formula \eqref{eq:canonicaldualautom} by $\ev_x \circ \sigma_{x^*,x}^{-1}$ and using the triangle identities yields $\lambda_x = \id_x$.
    Since the identity is iso-positive, we obtain (2).
\end{proof}

\begin{remark}
    Condition \eqref{eq:classicaldaggercpt} is sometimes stated with the dagger on the evaluation map instead.
    This is equivalent, as can be seen by taking the dagger of the diagram and using that the braiding is unitary.
\end{remark}

\begin{remark}
Keeping the work of Penneys in mind \cite{penneysUDF}, it can be opaque whether a definition in a symmetric monoidal dagger category involving a unitary dual functor depends on the choice of duals.
However, according to Definition \ref{def:moderndaggercpt}, it is a property of a symmetric monoidal dagger category to be dagger compact.
\end{remark}

\section{Fermionic dagger categories}
\label{sec:fermdagger}
\subsection{Super Hilbert spaces}
\label{sec:shilb}

Finite-dimensional super Hilbert spaces form one of the most physically relevant examples of dagger categories with duals.
Most of this section is well known, our main new insight being that super Hilbert spaces are `dagger compact up to fermion parity'.
We will call this notion `fermionically dagger compact', see Definition \ref{def:fermdagger}.

Recall that a super vector space is defined to be a $\Z/2$-graded (complex) vector space $V = V_0 \oplus V_1$.
We will denote the degree of a homogeneous vector $v \in V$ by $|v| \in \Z/2$ and the grading operator $v \mapsto (-1)^{|v|} v$ by $(-1)^F_V$, also called \emph{fermion parity}.
We will often refer to degree zero vectors as even and degree one vectors as odd.
A linear map $T: V \to W$ between super vector spaces is \emph{even} (or of degree zero) if $T(V_0) \subseteq W_0$ and $T(V_1) \subseteq W_1$.
Similarly, $T$ is \emph{odd} (or of degree one) if $T(V_0) \subseteq W_1$ and $T(V_1) \subseteq W_0$.
We denote the category of finite-dimensional super vector spaces with even linear maps by $\sVect$.
This category is symmetric monoidal with the obvious tensor product and the braiding
\[
v \otimes w \mapsto (-1)^{|v||w|} w \otimes v
\]
given by the Koszul sign rule, motivated by fermions in physics.

Similarly to finite-dimensional vector spaces, this category comes equipped with a canonical symmetric monoidal anti-involution $dV := \ol{V}^*$.
Here the dual $V^* = \Hom_{\Vect}(V,\C)$ of a super vector space $V$ is the vector space of all linear maps.
A functional $f \in V^*$ is of degree zero/one if the corresponding linear map $V \to \C$ is of degree zero/one, where $\C$ is put in even degree.
The dual functor on $\sVect$ associated to this choice of dual is given by
\begin{equation}
    \label{eq:dualoflinearmap}
T^*(f)(v) = f(Tv) \quad T\colon  V \to W, f \in W^*.
\end{equation}
We will identify $\ol{V}^*$ and $\ol{V^*}$ by defining $\ol{f}(\ol{v}) = \ol{f(v)}$.
We choose the monoidal data $\chi\colon  \ol{V}^* \otimes \ol{W}^* \cong \ol{V \otimes W}^*$ given by
\[
\chi(\ol{f} \otimes \ol{g})(\ol{v} \otimes \ol{w}) = (-1)^{|g||v|} \ol{f(v) g(w)} \quad f \in V^*, g \in W^*, v \in V, w \in W.
\]
The canonical isomorphism $\eta_V \colon  V \to d^2 V$ is explicitly given by mapping $v \in V$ to $\Phi_v$, which evaluates functionals at $v$:
\[
\Phi_v(f) := (-1)^{|f||v|} f(v).
\]
For more motivation and details on these sign choices, we refer the reader to Appendix \ref{app:boringsign}.

\begin{proposition}
\label{prop:sHerm}
    The Hermitian completion of the symmetric monoidal anti-involutive category $\sVect$ is the symmetric monoidal $\dagger$-category in which objects are pairs $(V,\langle.,.\rangle)$ consisting of a super vector space and a nondegenerate sesquilinear pairing
    \[
    \langle .,. \rangle\colon  V \times V \to \C
    \]
    such that $V_0$ and $V_1$ are orthogonal and
    \begin{equation}
    \label{eq:shermpairing}
    \langle v,w \rangle = (-1)^{|v||w|} \ol{\langle w,v \rangle}
    \end{equation}
    for all homogeneous $v,w \in V$.
    Morphisms between super vector spaces equipped with such Hermitian pairings $\langle .,. \rangle_V$ and $\langle .,. \rangle_W$ are even linear maps $T\colon  V \to W$.
    The dagger of such a linear map is the unique operator $T^\dagger\colon  W \to V$ such that
    \[
\langle Tv, w \rangle_{W} = \langle v, T^\dagger w \rangle_{V}
\]
for all $v \in V$ and $w \in W$.
The tensor product of Hermitian pairings $\langle .,. \rangle_V$ and $\langle .,. \rangle_W$ on super vector spaces $V$ and $W$ is the Hermitian pairing
    \[
    \langle v_1 \otimes w_1, v_2 \otimes w_2 \rangle := (-1)^{|v_2||w_1|} \langle v_1, v_2 \rangle_V \langle w_1, w_2 \rangle_W.
    \]
\end{proposition}
\begin{proof}
    A linear isomorphism $h\colon V \cong \ol{V}^*$ is equivalent to a nondegenerate sesquilinear pairing
    \[
    \langle .,. \rangle\colon  V \times V \to \C.
    \]
    The map $h$ preserves the grading if and only if $V_0$ and $V_1$ are orthogonal.
    The diagram
    \begin{equation}
\begin{tikzcd}
V \ar[r,"h"] \ar[d,"\eta_V"] & \ol{V}^* &
\\
\ol{\ol{V}^*}^* \ar[ru,"dh"']& 
\end{tikzcd}
\end{equation}
commutes if and only if \eqref{eq:shermpairing} holds
    for all homogeneous $v,w \in V$.
    The dagger $T^\dagger$ of a linear map $T\colon  V \to W$ between super vector spaces equipped with Hermitian pairings $h_V$ and $h_W$ respectively, is defined to be the composition
    \[
    W \xrightarrow{h_W} \ol{W}^* \xrightarrow{\ol{T}^*} \ol{V}^* \xrightarrow{h_V^{-1}} V.
    \]
    By direct computation one can verify that
    \begin{equation}
    \label{eq:superdagger}
\langle Tv, w \rangle_{W} = \langle v, T^\dagger w \rangle_{V}, 
\end{equation}
for all homogeneous $v \in V$ and $w \in W$.
Writing out the formula \ref{def:tensorproductofhermstrs} for the tensor product gives the monoidal structure on the Hermitian completion.
\end{proof}

\begin{remark}
    For an linear map of odd degree, there would be a Koszul sign in Equation \eqref{eq:superdagger}.
    This is however irrelevant for this paper, as we will only consider $\sVect$ to have even morphisms.
\end{remark}

Note that if $\langle.,.\rangle$ is a Hermitian pairing as in the above proposition and $v \in V$ is homogeneous, then
\[
\langle v,v \rangle = (-1)^{|v|} \overline{ \langle v,v \rangle}
\]
is real when $v$ is even and imaginary when $v$ is odd.
In defining positive definiteness, we decide to call imaginary numbers of the form $a \cdot i$ where $a \in \R_{>0}$ positive:

\begin{definition}
\label{def:shilbspace}
A Hermitian pairing on $V \in \sVect$ is \emph{positive (definite)} if for all even $v \in V$
\[
\langle v,v \rangle \geq 0
\]
and for all odd $v \in V$
\[
 \frac{\langle v,v \rangle}{i} \geq 0.
\]
We also say that $(V,h)$ is a (finite-dimensional) \emph{super Hilbert space}.
Let $\sHilb \subseteq \Herm(\sVect)$ be the symmetric monoidal dagger category on the monoidal positivity structure given by the super Hilbert spaces.
\end{definition}

\begin{remark}
    Somewhat surprisingly, a computation shows that the tensor product of two super Hilbert spaces is again a super Hilbert space.
    In other words, super Hilbert spaces indeed form a monoidal positivity structure.
\end{remark}

\begin{remark}
    We refer readers that are uneasy with the signs in Equation \eqref{eq:shermpairing} and the above discussion to Appendix \ref{app:boringsign}.
Our convention \ref{def:shilbspace} for super Hilbert spaces is not standard in the literature, however see \cite[Sign manifesto]{deligne1999quantum}, \cite[Section 12.5]{moore2014quantum} and \cite{freed1999five}.

\end{remark}

Let $\{e_1, \dots, e_n\} \subset V$ be a basis of homogeneous elements ordered by degree.
We can decompose a Hermitian pairing $h \colon V \to \ol{V}^*$ in components as $h(e_i)(\bar{e}_j) = h_{ij}$ so that $h_{ji} = (-1)^{|e_i||e_j|} \overline{h_{ij}}$. 
Note that $h_{ij} = 0$ if $e_i$ and $e_j$ are of different degree.
In matrix notation
\[
h =
\begin{pmatrix}
A & 0 \\
0 & B
\end{pmatrix},
\]
where $A$ is self-adjoint and $B$ is skew-adjoint.
The pairing $h$ is positive if and only if $A$ and $\frac{B}{i}$ are positive matrices in the ordinary sense.
We can diagonalize these matrices by a Gram-Schmidt process:

\begin{definition}
\label{def:signature}
An \emph{orthonormal basis} with respect to a Hermitian pairing on $V$ is a homogeneous basis $(e_1, \dots, e_n) \subseteq V$ such that 
\[
\langle e_i, e_j \rangle = \delta_{i j} \langle e_j,e_j \rangle,
\]
where $\langle e_j,e_j \rangle \in \{\pm 1, \pm i\}$ will be called the \emph{norm} of the orthonormal vector.
The quadruple $(p_1,p_2,p_3,p_4)$ containing the number of $j$s for which the norm $\langle e_j,e_j \rangle$ are respectively $+1,-1,+i$ and $-i$ is called the \emph{signature} of the pairing.
\end{definition}

It is easy to show the signature is well-defined.
Orthonormal bases are convenient for computations: the dagger of a matrix is the usual conjugate transpose whenever the norms in the domain and codomain agree, and minus the conjugate transpose whenever the norms disagree.



We discuss positivity structures on $\mathcal{C} = \sVect$.
We have an isomorphism of monoids $\pi_0(\mathcal{C}) \cong \mathbb{N} \times \mathbb{N}$ given by the dimensions of the even and odd parts respectively. 
Super Hermitian vector spaces are unitarily equivalent if and only if they have the same signature in the sense of Definition \ref{def:signature} and so we obtain $\pi_0^h \mathcal{C} \cong \mathbb{N} \times \mathbb{N} \times \mathbb{N} \times \mathbb{N}$. 
The map $\pi_0^h \mathcal{C} \to \pi_0 \mathcal{C}$ is given by
\[
\mathbb{N} \times \mathbb{N} \times \mathbb{N} \times \mathbb{N} \to \mathbb{N} \times \mathbb{N} \quad (p_1, p_2, p_3, p_4) \mapsto (p_1 + p_2, p_3 + p_4).
\]

\begin{remark}
Similarly to Example \ref{ex:hermstrsonsvec}, there are many pathological monoidal positivity structures on $\sVect$ and many of them are even minimal.
However, if we restrict to those notions that are preserved under direct sum, only very reasonable positivity structures such as super Hilbert spaces remain.
This might give a slightly more reasonable notion of `positivity structure' in the $\C$-linear setting, see \cite{henriques2017bicommutant} for a different approach to positivity.
Namely, the symmetric monoidal dagger subcategory of $\Herm(\sVect)$ is fixed once we decide on which Hermitian pairings we allow on the even and the odd line.
In particular, the $\dagger$-equivalence class of such a dagger category is determined by specifying whether we allow positive definite or negative definite pairings on the even part and odd part separately.
Therefore, there are two nontrivial symmetric monoidal full dagger subcategories $\mathcal{C} \subseteq \Herm(\sVect)$ closed under direct sum up to $\dagger$-equivalences $\mathcal{C} \cong \mathcal{C}'$ making the triangle commute.
They are given by $\sHilb$ and $\sHilb_{odd-neg}$, the symmetric monoidal dagger category of super Hermitian vector spaces of which the even part is positive definite and the odd part is negative definite.
Note that there is still a symmetric monoidal $\dagger$-equivalence $\sHilb \cong \sHilb_{odd-neg}$.
This equivalence covers the identity functor $F\colon  \sVect \to \sVect$ equipped with the anti-involutive datum $Fd \cong dF$ given by the fermion parity operator.
\end{remark}

One important datum that categories of fermionic nature come equipped with is a $B\Z/2$-action, which we suggestively denote $(-1)^F$.

\begin{definition}
\label{def:antiinvBAaction}
Let $A$ be an abelian group and $\mathcal{C}$ a symmetric monoidal category.
A \emph{$BA$-action} on $\mathcal{C}$ consists of a collection of monoidal natural automorphisms of the identity $a_c \in \Aut c$ for all $a \in A$ that satisfy $a_c \circ a_c' = (a a')_c$ for all $c \in \mathcal{C}$ and $a,a' \in A$.
If $\mathcal{C}$ is a symmetric monoidal anti-involutive category, the $BA$-action is \emph{anti-involutive} if $a_\bullet$ is an anti-involutive natural automorphism for all $a \in A$.
In other words, $(a^{-1})_{dc} = da_c$ for all $a \in A$ and $c \in \obj \mathcal{C}$.
\end{definition}

    In particular, a \emph{$B\Z/2$-action} is a monoidal natural involution $(-1)^F\colon  \id_\mathcal{C} \Rightarrow \id_{\mathcal{C}}$.
    Note that because $(-1)^F$ is both anti-involutive and an involution, it is also self-adjoint in the sense that $d(-1)^F_x = (-1)^F_{dx}$.

    \begin{example}
    $\sVect$ has a canonical symmetric monoidal $B\Z/2$-action given by mapping a vector space $V$ to the grading operator $(-1)^F_V\colon  V \to V$.
    It is anti-involutive for $V \mapsto \ol{V}^*$ because $\ol{(-1)^F_V} = (-1)^F_{\ol{V}}$.
\end{example}

If $\mathcal{C}$ is a symmetric monoidal anti-involutive category with monoidal positivity structure $P$ and anti-involutive $B\Z/2$-action $(-1)^F$, we denote by $\tilde{P}$ the collection of compositions
\[
x \xrightarrow{(-1)^F_x} x \xrightarrow{h} dx
\]
such that $h \in P$.
This composition is again a Hermitian pairing, because $(-1)^F_x$ is self-adjoint and natural.
We will now show that the dual involution of monoids
\[
(.)^*\colon  \pi_0^h(\sVect) \to \pi_0^h(\sVect)
\]
is given by $P \mapsto \tilde{P}$.

\begin{lemma}
\label{lem:hermstratdual}
Let $(V,h) \in \Herm(\sVect)$ be a super Hermitian vector space with signature $(p_1,p_2,p_3,p_4)$. 
Then the dual super Hermitian vector space $(h^*)^{-1}\colon  \overline{V^*} \to V^{**}$ has signature $(p_1,p_2,p_4,p_3)$.
\end{lemma}
\begin{proof}
Let $\{e_1, \dots, e_n\}$ be an orthonormal basis of $V$.
This induces a dual basis $\{e^1, \dots, e^n\}$ of $V^*$.
Note that
\[
h(e_i) = h(e_i)(e_i) e^i.
\]
Now note that if $f \in \ol{V}^*,g\in V^*$, then by formula \eqref{eq:dualoflinearmap} $h^{*-1}(f) \in V^{**}$ satisfies
\[
[h^{*-1}(f)](g) = [f \circ h^{-1}](g).
\]
In particular, we compute
\[
h^{*-1}(e^i)(e^i) = h(e_i)(e_i)^{-1}.
\]
Since we have $h(e_i,e_i)^{-1} = h(e_i,e_i)$ if $h(e_i,e_i) \in \{\pm 1\}$ and $h(e_i,e_i)^{-1} = - h(e_i,e_i)$ if $h(e_i,e_i) \in \{\pm i\}$, the result follows.
\end{proof}

\begin{corollary}
\label{cor:shilbfermdaggercpt}
    The standard dual functor on $\Herm(\sVect)$ maps $\sHilb \subseteq \Herm(\sVect)$ to $\sHilb_{odd-neg}$, the dagger category of super Hermitian vector spaces of which the even part is positive definite and the odd part is negative definite.
\end{corollary}

We see that $\sHilb$ is not dagger compact.
However, it still admits a unitary dual functor $\sHilb \to \sHilb^{\cop}$ given by twisting the standard unitary dual functor by $(-1)^F$.
Note that $\sHerm$ is dagger compact, but it also has this $(-1)^F$-twisted unitary dual functor. 
These two dual functors on $\sHerm$ are not unitarily equivalent.

\subsection{The Spin-Statistics Theorem}
\label{sec:spinstatistics}

The categories appearing in the study of unitary topological field theories with fermions turn out to have some abstract features in common.
First of all, categories such as $\sHilb$ and $\Bord^{\Spin}_d$ are symmetric monoidal dagger categories with duals and $B\Z/2$-action.
Also, their dagger category structure is most conveniently constructed by choosing Hermitian pairings for an anti-involution.
However, unlike $\Hilb$ and the oriented bordism dagger category, $\sHilb$ and the spin bordism dagger category are not dagger compact.
Namely, let $P$ denote the positivity structure on $\sVect$ so that $\sVect_P = \sHilb$.
Then the Corollary \ref{cor:shilbfermdaggercpt} says that the standard dual functor maps
\[
\sVect_P \to (\sVect_{\tilde{P}})^{\op}.
\]
This turns out to be a feature, not a bug, because it allows us to prove the Spin-Statistics Theorem.

For the rest of this section, $\mathcal{C}$ denotes a symmetric monoidal anti-involutive category with monoidal anti-involutive $B\Z/2$-action $(-1)^F$ and monoidal positivity structure $P$. 
Note that this induces a unitary monoidal $B\Z/2$-action on $\mathcal{C}_P$.
Separately, $\mathcal{D}$ will denote a symmetric monoidal dagger category with monoidal unitary $B\Z/2$-action $(-1)^F$.
We make a $(-1)^F$-twisted analogue of Definition \ref{def:moderndaggercpt}:

\begin{definition}
\label{def:fermdagger}
    The symmetric monoidal dagger category $\mathcal{C}_P$ is called \emph{fermionically dagger compact} (with respect to $(-1)^F$) if it has duals and the canonical dual functor on $\Herm \mathcal{C}$ restricts to a symmetric monoidal dagger functor
    \[
    \mathcal{C}_P \to \mathcal{C}_{\tilde{P}}^{\op}.
    \]
\end{definition}

Concretely, this means that for all $(h\colon x \to dx) \in P$, the composition
\[
        x^* \xrightarrow{(-1)^F_{x^*}} x^* \xrightarrow{h^{*-1}} (dx)^* \cong d(x^*)
        \]
        is again in $P$.
Analogously to Lemma \ref{lem:compactcondition}, it can be shown that being fermionically dagger compact is a well-defined property of a symmetric monoidal dagger category with unitary monoidal $B\Z/2$-action.

\begin{remark}
\label{rem:fermdaggercptclassical}
Similarly to Proposition \ref{prop:daggercptTFAE}, it can be shown that a symmetric monoidal dagger category $\mathcal{D}$ is dagger compact if and only if for every object $x \in \mathcal{D}$, there exists a duality $\ev_x\colon  x^* \otimes x \to 1$ such that
\begin{equation}
\label{eq:doubleofelbow}
    \begin{tikzcd}[column sep = 50]
        1 \ar[d,"\coev_x"] \ar[r,"\ev_x^\dagger"] & x^* \otimes x \ar[d,"\sigma_{x^*,x}"]
        \\
        x \otimes x^* \ar[r, "(-1)^F_x \otimes \id_{x^*}"] & x \otimes x^*
    \end{tikzcd}
    \end{equation}
commutes.
Thinking of $(-1)^F$ as the operation of `rotating by $2\pi$', we obtain a pictorial interpretation of diagram \ref{eq:doubleofelbow} in Figure \ref{fig:fermdagger}.
This string diagram suggests a graphical calculus for fermionically dagger compact categories analogous to Selinger's calculus for dagger compact categories \cite[Theorem 3.11]{selinger2007positive}, also see \cite[Theorem 7.9]{selinger2011survey}.
\end{remark}

\begin{figure}
    \centering
    \caption{String diagram illustration of Definition \ref{def:fermdagger} using Remark \ref{rem:fermdaggercptclassical}. }
    \label{fig:fermdagger}
\begin{tikzpicture}[scale=.2]
		\node  (5) at (1, 7) {};
		\node  (6) at (2, 7) {};
		\node  (7) at (2, 7) {};
		\node  (8) at (2, 6) {};
		\node  (10) at (5, 5) {};
		\node  (11) at (-2, 5) {};
		\node  (12) at (-2, -10) {};
		\node  (13) at (5, -10) {};
		\node  (14) at (-31, 5) {};
		\node  (15) at (-31, -10) {};
		\node  (16) at (-27, -3) {};
		\node  (17) at (-26, -3) {};
		\node  (18) at (-26, -2) {};
		\node  (19) at (-26, -3) {};
		\node  (20) at (-7.75, -2.5) {};
		\node  (21) at (-7.75, -2.5) {};
  \node (22) at (-7,-3) {$=$};
  \node  (23) at (-31, -12) {$x^*$};
  \node  (24) at (-31, 7) {$x$};
  \node (25) at (-4,5) {$x$};
  \node (26) at (-4,-10) {$x^*$};
  \node (27) at (6,10) {$(-1)^F_x$};
		\draw [bend left=15] (10.center) to (8.center);
		\draw [in=135, out=45, looseness=14.25] (6.center) to (5.center);
		\draw [bend right=15] (11.center) to (7.center);
		\draw [bend left=90, looseness=2.25] (10.center) to (13.center);
		\draw (13.center) to (12.center);
		\draw (14.center) to (17.center);
		\draw [bend left=15] (15.center) to (16.center);
		\draw [in=-60, out=45, looseness=78.75,overlay] (18.center) to (19.center);
\end{tikzpicture}
\end{figure}
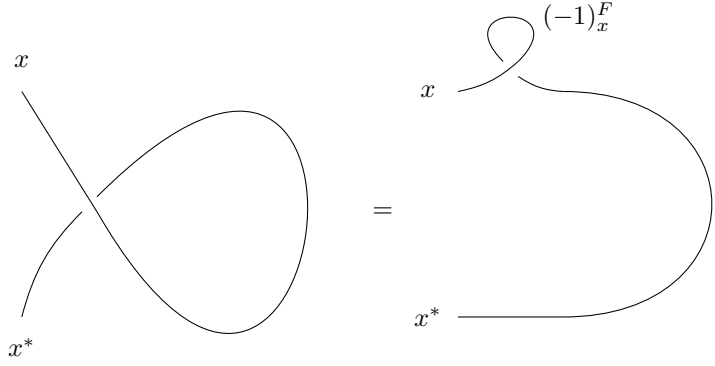

\begin{example}
    Let $\mathcal{D}$ be a symmetric monoidal dagger category, which we equip with the trivial $B\Z/2$-action. 
    Then $\mathcal{D}$ is fermionically dagger compact if and only if it is dagger compact.
\end{example}

\begin{example}
    It follows by Corollary \ref{cor:shilbfermdaggercpt} that $\sHilb$ is fermionically dagger compact.
\end{example}

\begin{remark}
Note that every fermionically dagger compact category $\mathcal{D}$ does have a canonical unitary dual functor 
\[
(.)^*_{(-1)^F}\colon  \mathcal{D} \to \mathcal{D}^{\op},
\]
given by twisting the standard unitary dual functor 
\[
(.)^*\colon  \mathcal{D} \to \mathcal{D}_{\tilde{P}}^{\op}
\]
with $(-1)^F$.
It can explicitly be constructing by picking duals on all objects satisfying \eqref{eq:doubleofelbow}.
\end{remark}

\begin{remark}
Recall that a pivotal structure on a monoidal category $\mathcal{D}$ with duals is a monoidal natural isomorphism $\phi:(.)^{**} \to \id_{\mathcal{D}}$.
Every symmetric monoidal category has a canonical pivotal structure using that $x^{**}$ and $x$ are both dual to $x^*$.
If $\mathcal{D}$ is a monoidal dagger category with a chosen unitary dual functor, the theory of \cite{penneysUDF} gives a corresponding unitary pivotal structure on $\mathcal{D}$.
In case $\mathcal{D}$ is dagger compact and we take the standard unitary dual functor, this will give the canonical pivotal structure $\phi$.
In case $\mathcal{D}$ is fermionically dagger compact and we consider the unitary dual functor of the previous remark, we obtain the pivotal structure $\phi'$ on $\mathcal{D}$ which differs from the pivotal structure $\phi$ with the braiding with $(-1)^F$. 
For example, if $\mathcal{D} = \sHilb$, the trace with respect to $\phi'$ is the usual ungraded trace
\[
\tr_{\phi'} T = \tr_{\phi} T (-1)^F = \tr_{s} T (-1)^F = \tr_{ungr} T.
\]
\end{remark}

The following lemma will imply our main theorem.

\begin{lemma}
\label{lem:prespinstatistics} 
    Let $\mathcal{D}_1,\mathcal{D}_2$ be fermionically dagger compact categories.
    Let $F\colon (\mathcal{C}_1)_{P_1} \to (\mathcal{C}_2)_{P_2}$ be a symmetric monoidal dagger functor.
    Then $F(\tilde{P_1}) \subseteq \tilde{P_2}$.
\end{lemma}
\begin{proof}
The argument is analogous to the argument in the proof of Lemma \ref{lem:daggercptwelldefd}. 
The symmetric monoidal dagger functor $F$ corresponds to a symmetric monoidal anti-involutive functor $F\colon  \mathcal{C}_1 \to \mathcal{C}_2$ which maps $P_1$ to $P_2$.
The canonical natural isomorphism $\xi_x\colon  F(x^*) \cong F(x)^*$ specifying uniqueness of duals provides an anti-involutive filling of the square
\[
 \begin{tikzcd}
        \mathcal{C}_1 \ar[r,"F"] \ar[d,"(.)^*"] &  \mathcal{C}_2 \ar[d,"(.)^*"] \ar[dl,Rightarrow, "\xi", shorten <=3, shorten >=3]
        \\
        (\mathcal{C}_1)^{\cop} \ar[r,"F"] &  (\mathcal{C}_2)^{\cop}
    \end{tikzcd}.
\]
by Lemma \ref{lem:lemma2}.
Applying Lemma \ref{lem:lemma1} to the case where $G_1(x) = F(x)^*$ and $G_2(x) = F(x^*)$, we obtain
\[
F(\tilde{P_1}) = F(P_1^*) = F(P_1)^* = \widetilde{F(P_1)}.
\]
Since $F(P_1) \subseteq P_2$, it follows that $F(\tilde{P_1}) \subseteq \tilde{P_2}$.
\end{proof}

\begin{theorem}[Main Theorem]
\label{mainth}
    Let $F\colon  \mathcal{D}_1 \to \mathcal{D}_2$ be a symmetric monoidal dagger functor between fermionically dagger compact categories such that every iso-positive involution in $\mathcal{D}_2$ is the identity.
    Then $F$ is $B\Z/2$-equivariant.
\end{theorem}
\begin{proof}
Let $d_1 \in \mathcal{D}_1$ be an object.
We claim that if $f^\dagger f \colon  F(d_1) \to F(d_1)$ is an iso-positive automorphism of an object in the image, then $f^\dagger f \circ (-1)^F_{F(d_1)} F((-1)^F_{d_1})$ is again iso-positive.
    Given the claim, the theorem follows by taking $f = \id_{d_1}$ to find that $j := (-1)^F_{F(d_1)} F((-1)^F_{d_1})$ is iso-positive.
    Indeed, by naturality of $(-1)^F_{\mathcal{D}_1}$, the map $j$ is also an involution, so that $j = \id_{F(d_1)}$ as desired. 
    
    To prove the claim, first consider anti-involutive categories $\mathcal{C}_1$ and $\mathcal{C}_2$ with positivity structures $P_1$ and $P_2$ respectively and let $F\colon  \mathcal{C}_1 \to \mathcal{C}_2$ be an anti-involutive functor that maps $P_1$ to $P_2$.
    We obtain from Lemma \ref{lem:prespinstatistics} that $F(\tilde{P_1}) \subseteq \tilde{P_2}$.
    Explicitly, this means that if $h\colon  c_1 \to dc_1$ is in $P_1$, then we get that
    \[
    F(c_1) \xrightarrow{(-1)^F_{c_1}} F(c_1) \xrightarrow{F((-1)^F_{c_1})} F(c_1) \xrightarrow{F(h)} F(dc_1) \cong dF(c_1)
    \]
    is in $P_2$.
    By naturality, the order of the $B\Z/2$-action on $\mathcal{C}_2$ and the $B\Z/2$-action on $\mathcal{C}_1$ pushed forward under $F$ is irrelevant in the above considerations.
    Taking $\mathcal{C}_i$ to be $\mathcal{D}_i$ and $P_i = \Pos(\mathcal{D}_i)$ the positivity structure from Example \ref{ex:notionofposofdaggercat} using Theorem \ref{th:symmetricjantheorem}, we obtain the claim.
\end{proof}

\begin{remark}
    Note that $\mathcal{D}_2 = \sHerm$ is an example of a fermionically dagger compact category which admits non-identity iso-positive involutions.
    It is straightforward to find examples of symmetric monoidal dagger functors from a fermionically dagger compact category to $\sHerm$ that are not $B\Z/2$-equivariant.
For a physically relevant example, consider the one-dimensional oriented bordism category $\Bord_1^{SO}$ with the trivial $B\Z/2$-action.
The functor $Z\colon  \Bord_1^{SO} \to \sHerm_\C$, which assigns the odd line with the negative definite inner product to a point, assembles into a symmetric monoidal dagger functor.
It is not $B\Z/2$-equivariant, because it does not factor through $\Herm_\C$. 
From the perspective of TQFT, this is a `integer spin fermionic theory', see \cite[Appendix E.1]{TachikawaYonekura} for a $1$-dimensional quantum field theory which recovers this TQFT as its low-energy effective field theory.
    This does not violate the Spin-Statistics Theorem, because this is only a Hermitian, not a unitary TQFT.
\end{remark}

\begin{remark}
\label{rem:braidedfunctorpreservesantiinvdata} 
At first sight, it seems as if we didn't use the fact that $F$ is symmetric in the proof of the Main Theorem.
However, we used the braiding to make the dual functor monoidal, instead of op-monoidal, see the discussion under Proposition \ref{prop:dualsuniqeu}.
Since the canonical anti-involutive data of the dual functor depends on the braiding, it is only preserved by symmetric functors.
\end{remark}

\begin{figure}
\centering
\includegraphics[width=0.6\textwidth]{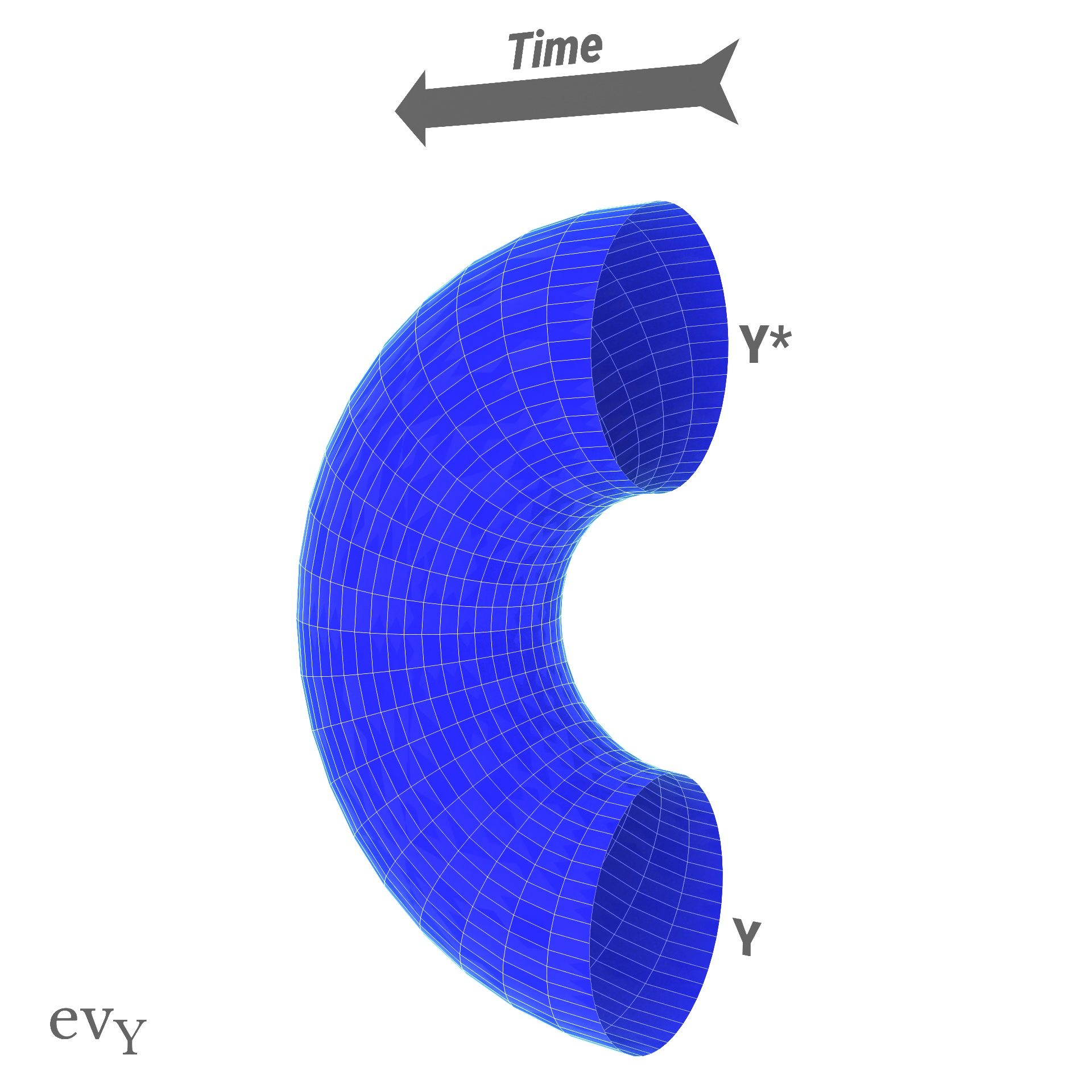}
  \caption{The macaroni bordism from the disjoint union of the space $Y$ and its dual $Y^*$. 
  It is the evaluation map $\ev_Y$ realizing $Y^*$ as the dual of $Y$ in the bordism category.}
  \label{fig:macaroni}
\end{figure}

    Let $\Bord$ denote a symmetric monoidal bordism category with $B\Z/2$-action suggestively denoted $Y \mapsto (-1)^{2s}_Y$.
    Our main theorem motivates us to construct fermionically dagger compact $\dagger$-structures on $\Bord$.
    We will not explain this in detail here, instead referring to \cite{mythesis}.
    As we need it for Corollary \ref{cor:spinstatistics}, we instead sketch how this works for the case where $\Bord = \Bord^{\Spin}_d$, following closely the treatment in \cite[Section 4]{freedhopkins}.
    Here $\Bord^{\Spin}_d$ is the spin bordism category equipped with its $B\Z/2$-action given by the spin flip automorphism $Y \mapsto (-1)^{2s}_{Y}$, as explained in the introduction.

    Recall that a spin structure $(P,\alpha)$ on a $d$-dimensional vector bundle $E \to X$ consists of a $\Spin_d$-principal bundle $P\to X$ and an isomorphism $\alpha: P \times_{\Spin_d} \R^d \cong E$ of real vector bundles, where $\Spin_d$ acts on $\R^d$ via the vector representation.
We model the bordism category to have objects $(d-1)$-dimensional manifolds $Y$ with $\Spin_d$-structure $(P,\alpha)$ on $TY \oplus \underline{\R}$.
We define the \emph{dual} of $(Y,P,\alpha)$ as $(Y,P,\alpha)^* := (Y,P,\id_{TY} \oplus -\id_{\underline{\R}})$, which we think of as `$Y$ with normal vector field pointing in the other direction'.
Morphisms $(Y_1,P_1,\alpha_1) \to (Y_2,P_2,\alpha_2)$ in $\Bord^{\Spin}_d$ consist of a $d$-dimensional manifold $W$ with boundary $Y_1 \sqcup Y_2$ and a spin structure on $TW$ which is compatible with the boundary spin structure $(Y_1,P_1,\alpha_1)^* \sqcup (Y_2,P_2,\alpha_2)$.
We then have to quotient by spin diffeomorphisms of $W$ compatible with the respective boundaries.
It is straightforward to show that $(Y,P,\alpha)^*$ is dual to $(Y,P,\alpha)$ using the macaroni bordism Figure \ref{fig:macaroni} as evaluation.

The \emph{orientation-reversal} involution $\beta: \Bord^{\Spin}_d \to \Bord^{\Spin}_d$ is defined as follows, see \cite[Section 4.1]{freedhopkins} and \cite[Section 4.5]{mythesis}.
Since the spin group sits inside the pin group as an index two subgroup, the associated $\Pin_d^+$-bundle $P \times_{\Spin_d} \Pin^+_d$ decomposes into two $\Spin_d$-bundles
\[
P \times_{\Spin_d} \Pin^+_d = (P \times_{\Spin_d} \Spin_d ) \sqcup (P \times_{\Spin_d} (\Pin^+_d \setminus \Spin_d)),
\]
the first being isomorphic to $P$.
Define $\beta(M,P,\alpha) := (M, P \times_{\Spin_d} (\Pin^+_d \setminus \Spin_d), \beta(\alpha))$, where $\beta(\alpha)(p,g,v) := \alpha(p,\rho(g)v)$.
Here $p \in P, g \in \Pin^+_d \setminus \Spin_d, v \in \R^n$ and $\rho: \Pin^+_d \to GL_d(\R)$ is the vector representation.
It is straightforward to make this $\Z/2$-action compatible with boundaries of bordisms.

Given the involution $\beta$, we can define the anti-involution $dY := \beta(Y)^*$, see \cite[Section 2.2]{mythesis} for details on the relationship between involutions and anti-involutions in symmetric monoidal categories.
We can define Hermitian pairings for this anti-involution as explained in \cite[Proposition 4.8]{freedhopkins}, see \cite[Section 4.6]{mythesis} for a general discussion of Hermitian pairings on bordism categories.
In more detail, first note that if $(Y,P,\alpha) \in \Bord^{\Spin}_d$, we can assume without loss of generality that the spin structure is `constant in the time direction'.
More precisely, $P = Q \times_{\Spin_{d-1}} \Spin_d$ and $\alpha(q,g,v) = \beta(q, \rho(g) v)$, where $(Q,\beta)$ is a $\Spin_{d-1}$-structure on $TY$.
We can then define\footnote{This map needs to be modified in case $d = 1$, see \cite[discussion above Lemma 4.6.11.]{mythesis}.} the map $h: (M, P,\alpha) \to d(M, P,\alpha)$ as the identity on $M$ and $h(q,g) := (q,e_n g)$ on $(q,g) \in Q \times_{\Spin_{d-1}} \Spin_d = P$.
It can be checked that these define Hermitian pairings and so we can give $\Bord^{\Spin}_d$ the dagger structure obtained from the positivity structure given by these Hermitian pairings on all objects as in Definition \ref{def:positivitystructure}.
This positivity structure is clearly monoidal with respect to disjoint union on the bordism category, making $\Bord^{\Spin}_d$ into a symmetric monoidal dagger category.
The $B\Z/2$-action on $\Bord^{\Spin}_d$ is unitary.

The condition of $\Bord^{\Spin}_d$ being fermionically dagger compact has a convenient geometric interpretation.
    Namely, note that by Remark \ref{rem:fermdaggercptclassical}, we have to check whether for every spatial slice $Y$ the macaroni bordism $\ev_Y\colon  Y^* \sqcup Y \to 1$ satisfies that $\ev_Y^\dagger$ differs from $\coev_Y$ with $(-1)^{2s}_Y$ on one of the legs.
The main input to show this is the computational fact that the element $\hat{h}_d := e_d \in \Pin^+_d$ anti-commutes with elements of $\Pin^+_{d-1} \setminus \Spin_{d-1}$, see \cite[Theorem 2.2.25]{mythesis} for details.

\begin{remark}
        The discussion above implies that in the spin bordism dagger category the double of the manifold with boundary $\ev_Y$, defined as the closed manifold $\ev_Y \circ \ev_Y^\dagger$, is $Y \times S^1_{ap}$.
        Here $S^1_{ap}$ is the anti-periodic spin circle, defined by the mapping torus of the spin flip automorphism on a point.
    Note that the manifold $Y \times S^1_{ap}$ bounds $Y \times D^2$, compare \cite[Corollary 7.52]{freed2019lectures} and the proof of \cite[Theorem 11.3]{freedhopkins}.
\end{remark}

\begin{corollary}
\label{cor:spinstatistics}
    Any symmetric monoidal dagger functor $F\colon  \Bord^{\Spin}_d \to \sHilb$ is $B\Z/2$-equivariant.
\end{corollary}
\begin{proof}
It follows by Corollary \ref{cor:shilbfermdaggercpt}, that $\sHilb$ is fermionically dagger compact.
Now let $j$ be an iso-positive involution in $\sHilb$.
Picking an orthonormal basis, we obtain a corresponding matrix $A$.
It is straightforward to show that $A$ is a positive matrix in the usual sense.
Since $A$ is an involution, we obtain that $A$ is the identity matrix.
    Since $\Bord^{\Spin}_d$ is fermionically dagger compact, the result follows by Theorem \ref{mainth}.
\end{proof}

\begin{remark}
    The results of \cite{freedhopkins} apply to the bordism category $\Bord^H_d$ of manifolds with certain more general $H$-structures.
    More precisely, these $H$ are the spacetime structure groups constructed from internal fermionic symmetry groups $G$, see \cite[Section 3.2]{stehouwermorita} for one formulation.
    Corollary \ref{cor:spinstatistics} holds more generally for symmetric monoidal dagger functors
    \[
    F\colon  \Bord^H_d \to \sHilb.
    \]
    In particular, it follows that every symmetric monoidal dagger functor
    \[
    F\colon  \Bord^{SO}_d \to \sHilb
    \]
    lands in $\Hilb$.
\end{remark}

\begin{remark}
Many proofs of the Spin-Statistics Theorem, such as in \cite[Volume 1, Part 2, \S 1.4] {deligne1999quantum} and \cite{baezspinstatistics}, crucially involve a rotation $PT$ between time and some direction in space.
The crux of the proof is then how a rotation squares to fermion parity. 
Depending on the approach, this can come in analogously into showing the spin bordism dagger category is fermionically dagger compact.
Namely, the definition of the dual in the bordism category as explained above involves reflecting in the time direction, while the `bar' $\beta$ can alternatively be defined by reflecting the bordism along an auxiliary direction, see \cite[Section 2.2]{yonekura2019cobordism} for details on this approach.
The Hermitian pairings $h_Y\colon  Y \cong \beta Y^*$ are then defined by rotating between this reflection in space and time.
The fact that $h_{Y^*}$ and the dual of $h_{Y}$ in the sense of Definition \ref{def:dualhermstr} differ by $(-1)^{2s}_Y$, is a consequence of the fact that in $\Spin(d)$ the lift of a rotation by $180^\circ$ squares to $(-1)^{2s}$.
\end{remark}

\appendix

\input{appendix}

\subsection*{Declarations}

\subsubsection*{Funding}

This article is based on the author's PhD thesis, written at the Max Planck Institute for Mathematics in Bonn and defended at the University of Bonn.
During the writing of the article, the author was employed by Dalhousie University and financially supported by the AARMS and the SCGCS.

\subsubsection*{Data availability}

This work does not generate any datasets. One
can obtain the relevant materials from the references below.

\subsubsection*{Conflict of interest}

The author has no competing interests to declare that are relevant to the content of this article.

\bibliography{biblio}{}
\bibliographystyle{plain}

\end{document}

%% file: preamble.tex
\usepackage[utf8]{inputenc}

\usepackage{amsmath}
\usepackage{amsfonts}
\usepackage{amssymb}
\usepackage{amsthm}
\usepackage{enumerate}
\usepackage{color}
\usepackage{stmaryrd}
\usepackage{slashed}
\usepackage{setspace}
\usepackage{verbatim} 

\usepackage{relsize}
\usepackage{tikz-cd}
\usepackage{tikz}
\usetikzlibrary{decorations.markings}
\usetikzlibrary{decorations.pathmorphing}
\usetikzlibrary{cd}
\usepackage{comment}
\usepackage{wrapfig}

\usepackage{geometry}
\usepackage{hyperref}
\usepackage{epigraph} 

\newcommand{\Bord}{\operatorname{Bord}}

\newcommand{\Aut}{\operatorname{Aut}}
\newcommand{\Hom}{\operatorname{Hom}}

\newcommand{\pt}{\operatorname{pt}}

\newcommand{\Spin}{\operatorname{Spin}}
\newcommand{\Pin}{\operatorname{Pin}}
\newcommand{\id}{\operatorname{id}}
\newcommand{\Herm}{\operatorname{Herm}}

\newcommand{\sVect}{\operatorname{sVect}}
\newcommand{\sHerm}{\operatorname{sHerm }}
\newcommand{\Hilb}{\operatorname{Hilb}}
\newcommand{\sHilb}{\operatorname{sHilb}}
\newcommand{\Vect}{\operatorname{Vect}}
\newcommand{\ev}{\operatorname{ev}}
\newcommand{\coev}{\operatorname{coev}}
\newcommand{\obj}{\operatorname{obj}}

\newcommand{\tr}{\operatorname{tr}}
\newcommand{\Fun}{\operatorname{Fun}}

\newcommand{\ol}{\overline}

\newcommand{\Z}{\mathbb{Z}}
\newcommand{\N}{\mathbb{N}}

\newcommand{\C}{\mathbb{C}}
\newcommand{\R}{\mathbb{R}}

\newcommand{\Cat}{\operatorname{Cat}}
\newcommand{\aICat}{\operatorname{aICat}}

\newcommand{\op}{\mathrm{op}}

\newcommand{\cop}{\mathrm{op}}

\newcommand{\Pos}{\operatorname{Pos}}

\newtheorem{theorem}{Theorem}[section]
\newtheorem{proposition}[theorem]{Proposition}
\newtheorem{lemma}[theorem]{Lemma}

\newtheorem{corollary}[theorem]{Corollary}

\makeatletter
\newtheorem*{rep@theorem}{\rep@title}
\newcommand{\newreptheorem}[2]{%
\newenvironment{rep#1}[1]{%
 \def\rep@title{#2 \ref{##1}}%
 \begin{rep@theorem}}%
 {\end{rep@theorem}}}
\makeatother

\newreptheorem{theorem}{Theorem}
\newreptheorem{proposition}{Proposition}
\newreptheorem{corollary}{Corollary}

\theoremstyle{definition}

\newtheorem{definition}[theorem]{Definition}

\theoremstyle{remark}

\newtheorem{remark}[theorem]{Remark}
\newtheorem{example}[theorem]{Example}

%% file: acknowledge.tex
\subsection*{Acknowledgements}

First and foremost, I thank my advisor Pete Teichner for the scrutiny of my PhD thesis and subsequent suggestions that led to this article.
I am grateful to Lukas M\"uller and Stephan Stolz for countless discussions on bordism theory, Hermitian pairings and dagger categories.
I want to thank all coauthors of the workshop on 
\href{http://categorified.net/dagger2023.html}{dagger higher categories} for their interest in generalizing my work to higher categories.
My special thanks go to David Reutter for having the patience to explain his ideas, and Theo Johnson-Freyd for bringing the group together and his feedback on my work.
I would like to mention Dan Freed as a major inspiration for all of my work.
Finally I thank two anonymous referees for finding several mistakes and improving the overall exposition.

I am grateful to the Max Planck Institute for Mathematics in Bonn where I wrote my PhD thesis on which this article is based.
I would like to thank Dalhousie University for providing the facilities to carry out my research.

%% file: appendix.tex
\section{Symmetric monoidal coherent dagger categories}
\label{app:symmondagger}

In this appendix, we will show the symmetric monoidal analogues of the results in \cite{jandagger}.
From an abstract perspective, note that the Cartesian product of categories makes the $2$-category of categories into a symmetric monoidal $2$-category.
Similarly to how monoids ($\mathbb{E}_1$-objects) in this category are monoidal categories, monoidal dagger categories are defined as monoids in the $2$-category of dagger categories.
Therefore, all we would have to do to understand symmetric monoidal dagger categories, is to understand how our constructions behave with respect to Cartesian products of categories.
However, we will take a more pedestrian approach.

In the first part of this section we work monoidally, postponing braidings to the end of the section.
If $\mathcal{C}$ is a monoidal category, $\mathcal{C}^{\op}$ will denote the monoidal category in which we reverse the direction of composition, but not of the tensor product.
As in the main text, a monoidal anti-involution is a fixed point under the action $\mathcal{C} \mapsto \mathcal{C}^{\cop}$ on monoidal categories.
We also provide definitions of monoidal dagger categories that are not symmetric.

\begin{definition}
An anti-involution $(d,\eta)$ on a monoidal category is called \emph{monoidal} if $d$ comes equipped with the data of being a monoidal functor $\mathcal{C} \to \mathcal{C}^{\cop}$ and $\eta$ is a monoidal natural transformation.
A \emph{monoidal anti-involutive functor} is an anti-involutive functor, which is also a monoidal functor such that the natural isomorphism $F \circ d \cong d \circ F$ of functors $\mathcal{C}_1 \to \mathcal{C}_2^{\op}$ is monoidal.
A monoidal anti-involutive natural transformation is a natural transformation that is both anti-involutive and monoidal.
We denote the $2$-category of monoidal anti-involutive categories by $\aICat_{\mathbb{E}_1}$.
A \emph{monoidal dagger category} is a dagger category that is also a monoidal dagger category such that $\otimes$ is a $\dagger$-functor and the unitor and the associator are unitary.
A monoidal functor between monoidal dagger categories is called a \emph{monoidal dagger functor} if it is a dagger functor and the isomorphisms $\mu_{c_1,c_2}\colon  F(c_1) \otimes F(c_2) \to F(c_1 \otimes c_2)$ and $1_{\mathcal{C}_2} \to F(1_{\mathcal{C}_1})$ are unitary.
A \emph{monoidal isometric transformation} between monoidal dagger functors is a natural transformation that is both isometric and monoidal.
Let $\Cat^\dagger_{\mathbb{E}_1}$ denote the $2$-category of monoidal dagger categories.
\end{definition}

Our first goal is to make $\Herm \mathcal{C}$ into a monoidal dagger category given a monoidal anti-involutive category $\mathcal{C}$.

\begin{proposition}
\label{prop:monoidalherm}
Let $\mathcal{C}$ be a monoidal anti-involutive category.
The tensor product of Hermitian pairings given in Definition \ref{def:tensorproductofhermstrs} gives $\Herm \mathcal{C}$ the structure of a monoidal dagger category.
\end{proposition}
\begin{proof}
We define the tensor product $(c_1, h_1) \otimes (c_2,h_2)$ by Definition \ref{def:tensorproductofhermstrs}, which still makes sense without a braiding on $\mathcal{C}$.
We start by verifying that this defines a Hermitian pairing on $c_1 \otimes c_2$.
For this we have to show the diagram
\[
\begin{tikzcd}
d(c_1 \otimes c_2)   & d c_1 \otimes d c_2 \arrow[l,"\chi_{c_1,c_2}"]   & c_1 \otimes c_2 \arrow[l,"{h_1 \otimes h_2}", swap] \arrow[dd,"\eta_{c_1 \otimes c_2}"] \arrow[dl,"{\eta_{c_1} \otimes \eta_{c_2}}"]
\\
\ & d^2 c_1 \otimes d^2 c_2 \arrow[dl,"{\chi_{dc_1,dc_2}}"] \arrow[u,"{d h_1 \otimes d h_2}"] & \ 
\\
d(dc_1 \otimes dc_2) \arrow[uu,"{d(h_1 \otimes h_2)}"]   & \ & d^2(c_1 \otimes c_2)  \arrow[ll,"{d\chi_{c_1,c_2}}"]
\end{tikzcd}
\]
commutes. 
This follows because $d$ and $\eta$ are monoidal and $h_1$ and $h_2$ are Hermitian pairings.
We take the associator $(c_1,h_1) \otimes ((c_2,h_2) \otimes (c_3, h_3)) \to ((c_1,h_1) \otimes (c_2,h_2)) \otimes (c_3, h_3)$ to be the associator $(c_1 \otimes c_2) \otimes c_3 \cong c_1 \otimes (c_2 \otimes c_3)$ of $\mathcal{C}.$
The pentagon identity follows immediately from the pentagon identity in $\mathcal{C}$.
Because $\eta$ is monoidal, the morphism $\eta(1)\colon  1 \to d^2(1)$ is equal to $du \circ u^{-1}$, where $u\colon  d(1) \to 1$ is the unitality data of $d$.
We see that $u$ is a Hermitian pairing on $1$ and so $(1,u) \in \Herm \mathcal{C}$ is a monoidal unit with the unitors equal to those of $\mathcal{C}$.

To show this makes $\Herm \mathcal{C}$ into a monoidal dagger category, we have to prove that $(f_1 \otimes f_2)^\dagger = f_1^\dagger \otimes f_2^\dagger$ and that the unitors and the associator are unitary.
Let $f_1\colon  (c_1,h_1) \to (c_1', h_1')$ and $f_2\colon  (c_2,h_2) \to (c_2',h_2')$ be morphisms in $\Herm \mathcal{C}$.
To obtain $(f_1 \otimes f_2)^\dagger = f_1^\dagger \otimes f_2^\dagger$, we note the following diagram commutes because $d$ is monoidal:
\[
\begin{tikzcd}[column sep=40]
d(c_1 \otimes c_2) & d(c_1' \otimes c_2') \arrow[l,"{d(f_1 \otimes f_2)}"]
\\
dc_1 \otimes dc_2 \arrow[u] & dc_1' \otimes dc_2' \arrow[u] \arrow[l,"df_1 \otimes df_2"]
\\
c_1 \otimes c_2 \arrow[u,"{h_1 \otimes h_2}"] & c_1' \otimes c_2' \arrow[u,"{h_1' \otimes h_2'}"] \arrow[l,"{f_1^\dagger \otimes f_2^\dagger}"]
\end{tikzcd}.
\]
For verifying the left unitor $\lambda_c\colon  (1, u) \otimes (c,h) \to (c,h)$ is unitary, we note the following diagram commutes
\[
\begin{tikzcd}[column sep =30]
d(1 \otimes c) & dc \arrow[l, "d(\lambda_c)"] 
\\
d1 \otimes dc \arrow[u] & \ 
\\
1 \otimes dc \arrow[u, "u \otimes \id"] \arrow[ruu,"{\lambda_{d(c)}}", swap]& \ 
\\
1 \otimes c \arrow[u,"{\id \otimes h}"] \arrow[r, "\lambda_c"]& c\arrow[uuu, "h", swap]
\end{tikzcd}.
\]
This follows by naturality of $\lambda$ and the unit condition on the functor $d$.
The right unitor is analogous.
Showing the associator is unitary is equivalent to showing the following commutes
\[
\begin{tikzcd}
d(c_1 \otimes (c_2 \otimes c_3)) \arrow[r,"d\alpha"] &  d((c_1 \otimes c_2) \otimes c_3)
\\
dc_1 \otimes d(c_2 \otimes c_3) \arrow[u]& d(c_1 \otimes c_2) \otimes dc_3\arrow[u]
\\
dc_1 \otimes (dc_2 \otimes dc_3) \arrow[u]& (dc_1 \otimes dc_2) \otimes dc_3 \arrow[l, "\alpha"]\arrow[u]
\\
c_1 \otimes (c_2 \otimes c_3) \arrow[u,"{h_1 \otimes (h_2 \otimes h_3)}"] & (c_1 \otimes c_2) \otimes c_3 \arrow[u,"{(h_1 \otimes h_2) \otimes h_3}"] \arrow[l, "\alpha"]
\end{tikzcd}.
\]
It does by associativity of the monoidal functor $d$ and the associator being natural.
\end{proof}

Next, we provide $\Herm$ on $1$-morphisms of $\aICat_{\mathbb{E}_1}$:

\begin{lemma}
For a monoidal anti-involutive functor, $\Herm F \colon  \Herm \mathcal{C}_1 \to \Herm \mathcal{C}_2$ is a monoidal dagger functor.
\end{lemma}
\begin{proof}
We saw in Proposition \ref{prop:monoidalherm} that $u_{\mathcal{C}_i}\colon  1_{\mathcal{C}_i} \to d1_{\mathcal{C}_i}$ defines a Hermitian pairing on $1_{\mathcal{C}_i}$, which made it into the monoidal unit of $\Herm \mathcal{C}_i$.
We have to show that $\epsilon\colon  1_{\mathcal{C}_2} \to \Herm F(1, u_{\mathcal{C}_2})$ is unitary.
Writing out the definition is exactly diagram \ref{eq:monoidalunital} in the definition of a monoidal anti-involutive functor.

Let $(c,h),(c',h') \in \Herm \mathcal{C}_1$.
We want to show that the monoidal data $\mu_{c,c'}\colon  \Herm F(c,h) \otimes \Herm F(c',h') \to \Herm F((c,h) \otimes (c',h'))$ is a unitary isomorphism.
Writing out the Hermitian pairings on the objects involved, this boils down to showing the following diagram commutes.
\[
\begin{tikzcd}[column sep = 40]
F(c) \otimes F(c') \arrow[r,"{F(h) \otimes F(h')}"] \arrow[d,"{\mu_{c,c'}}"] & F(dc) \otimes F(dc') \ar[r,"{\phi_c \otimes \phi_{c'}}"] \ar[d,"{\mu_{dc,dc'}}"] & dF(c) \otimes dF(c') \ar[r,"{\chi_{F(c) \otimes F(c')}}"] & d(F(c) \otimes F(c'))
\\
F(c \otimes c') \arrow[r,"F(h \otimes h')"] & F(dc \otimes dc') \arrow[r,"{F(\chi_{c,c'})}"] & F(d(c \otimes c')) \arrow[r,"{\phi_{c\otimes c'}}"] & dF(c\otimes c') \arrow[u,"{d \mu_{c,c'}}"]
\end{tikzcd}.
\]
The left square commutes because $F$ is a monoidal functor.
The right rectangle commutes by definition of $F$ being a monoidal anti-involutive functor, see Diagram \eqref{eq:monoidalantiinvfunctor}.
\end{proof}

\begin{corollary}
\label{cor:monoidalHerm2functor}
$\Herm$ is a $2$-functor
\[
\aICat_{\mathbb{E}_1} \to \Cat^\dagger_{\mathbb{E}_1}.
\]
\end{corollary}
\begin{proof}
    Monoidal anti-involutive natural transformations are defined as natural transformations that are both anti-involutive and monoidal and similar for monoidal isometric transformations on monoidal dagger categories.
Therefore we have already defined $\Herm$ on $2$-morphisms of $\aICat_{\mathbb{E}_1}$.
Also, we still have $\Herm(F_1 \circ F_2) = \Herm(F_1) \circ \Herm(F_2)$ as monoidal dagger functors, because the monoidal data $\mu_1, \mu_2$ of $\Herm(F_1)$ and $\Herm(F_2)$ are simply given by the monoidal data of $F_1$ and $F_2$.
The result now follows from Theorem \ref{th:jantheorem}.
\end{proof}

Note that Definition \ref{def:monoidalposstr} of a monoidal positivity structure still makes sense without a braiding. 

\begin{theorem}
\label{th:monoidaljanth}
    $\Herm$ is adjoint to the $2$-functor
    \[
    \Cat^\dagger_{\mathbb{E}_1} \to \aICat_{\mathbb{E}_1}.
    \]
    It induces an equivalence 
    \[
    \Cat^\dagger_{\mathbb{E}_1} \simeq (\aICat_{\mathbb{E}_1})_P,
    \]
    between monoidal dagger categories and monoidal anti-involutive categories with monoidal positivity structure. 
\end{theorem}
\begin{proof}
We first review the unit and counit of the adjunction appearing in the nonmonoidal case \cite[Theorem 4.9]{jandagger}.
The unit is given by sending the dagger category $\mathcal{D}$ to the dagger functor
\begin{equation}
    U_\mathcal{D}\colon  \mathcal{D} \to \Herm \mathcal{D},
\end{equation}
which includes $\mathcal{D}$ into its Hermitian completion by taking the Hermitian pairing on $x \in \mathcal{D}$ to be $h = \id_x\colon  x \to x^\dagger = x$.
The counit is given by sending the anti-involutive category $\mathcal{C}$ to the anti-involutive functor 
\begin{equation}
    K_{\mathcal{C}}\colon  \Herm \mathcal{C} \to \mathcal{C},
\end{equation}
which forgets the Hermitian pairing and has anti-involutive data $h_\bullet\colon  K_\mathcal{C} \circ \dagger \cong d \circ K_{\mathcal{C}}$ given by $(x,h) \mapsto h\colon  x \to dx$.

    We now make these into monoidal functors in the case where $\mathcal{C}$ is a monoidal anti-involutive category.
    The anti-involutive functor $K_\mathcal{C}\colon  \Herm \mathcal{C} \to \mathcal{C}$ has tautological monoidal data, which is associative.
    The fact that the natural isomorphism $h_\bullet\colon  K_{\mathcal{C}} \circ \dagger \cong d \circ K_{\mathcal{C}}$ is monoidal follows by the formula for the tensor product of two Hermitian pairings.
    Note that also that $K_{\mathcal{C}}$ is fully faithful and essentially surjects onto the objects that admit a Hermitian pairing.
    
    Now let $\mathcal{D}$ be a monoidal dagger category.
    We make $U_\mathcal{D}\colon  \mathcal{D} \to \Herm \mathcal{D}$ into a monoidal functor with the identity monoidal data which is natural and associative, also see Example \ref{ex:monoidaldaggerpositivity}.
    The identity is clearly also a unitary isomorphism
    \[
    (x_1, \id_{x_1}) \otimes (x_2, \id_{x_2}) \cong (x_1 \otimes x_2, \id_{x_1 \otimes x_2})
    \]
    in $\Herm \mathcal{D}$.
    Note that $U$ is strictly unital.

    Next, we show that $U$ is still natural in $\mathcal{D}$ and $K$ is natural in $\mathcal{C}$.
    For this, we only have to note that $K$ intertwines monoidal anti-involutive functors $\mathcal{C}_1 \to \mathcal{C}_2$ with their induced monoidal dagger functors $\Herm \mathcal{C}_1 \to \Herm \mathcal{C}_2$.
    The fact that $U$ and $K$ still define a strict $2$-adjunction in this monoidal scenario follows and so we proved the first statement.

    For the second statement, we recall that in the non-monoidal case, $U_\mathcal{D}$ gives an equivalence between the dagger category $\mathcal{D}$ and the dagger subcategory of $\Herm \mathcal{D}$, given by the positivity structure explained in Example \ref{ex:notionofposofdaggercat}.
    This positivity structure is monoidal as mentioned in Example \ref{ex:monoidaldaggerpositivity}, which also works without the braiding on $\mathcal{C}$.
    Since we already proved that $U_\mathcal{D}$ is a monoidal dagger functor and $U_\mathcal{D}$ restricts to an equivalence of dagger categories which is still a monoidal functor, it restricts to an equivalence of monoidal dagger categories.
    If $P$ is a monoidal positivity structure on $\mathcal{C}$ it induces a monoidal positivity structure on $\Herm \mathcal{C}$.
    Moreover, $K_\mathcal{C}$ preserves these positivity structures.
    Therefore if $P$ is monoidal, $K_{\mathcal{C}}$ induces an equivalence of monoidal anti-involutive categories equipped with monoidal positivity structures.
\end{proof}

We finish with the braided and symmetric situations, which are easier because the coherence data is on higher morphisms.
If $\mathcal{C}$ is braided by $\beta$, we give $\mathcal{C}^{\op}$ the braiding $\beta^{\op}_{c_1,c_2} := \beta^{-1}_{c_1,c_2}$.
We list the straightforward generalizations to the braided setting for reference.

\begin{definition}
    A \emph{braided monoidal dagger category} is a monoidal dagger category equipped with a unitary braiding.
    A \emph{braided monoidal dagger functor} is a braided monoidal functor which is a monoidal dagger functor.
    A \emph{braided monoidal isometric natural transformation} is a monoidal natural transformation that is objectwise an isometry.
    Let $\Cat^\dagger_{\mathbb{E}_2}$ denote the $2$-category of braided monoidal dagger categories.
A monoidal anti-involution $(d,\eta, \chi)$ on the braided category $\mathcal{C}$ is \emph{braided} if $d\colon  \mathcal{C} \to \mathcal{C}^{\cop}$ is a braided functor.
A \emph{braided anti-involutive functor} is a monoidal anti-involutive functor for which the underlying functor is braided.
A \emph{braided anti-involutive natural transformation} between braided anti-involutive categories is simply a monoidal anti-involutive natural transformation.
We denote the $2$-category of braided monoidal anti-involutive categories by $\aICat_{\mathbb{E}_2}$.
\end{definition}

\begin{lemma}
\label{lem:braidedherm}
If $\mathcal{C}$ is a braided anti-involutive category, then $\Herm \mathcal{C}$ is braided.
If $\mathcal{C}$ is symmetric, then so is $\Herm \mathcal{C}$.
\end{lemma}
\begin{proof}
As before, we need to define $(c_1, h_1) \otimes (c_2, h_2) \cong (c_2,h_2) \otimes (c_1,h_1)$ by the braiding of $\mathcal{C}$ and show it is a unitary isomorphism.
For this we have to show the diagram
\[
\begin{tikzcd}[column sep = 40]
d(c_1 \otimes c_2) & d(c_2 \otimes c_1) \arrow[l,"{d(\beta_{c_1, c_2})}", swap]
\\
dc_1 \otimes dc_2 \arrow[u] \arrow[r,"{\beta_{dc_1, dc_2}}"] & dc_2 \otimes dc_1 \arrow[u]
\\
c_1 \otimes c_2 \arrow[u,"h_1 \otimes h_2"] \arrow[r,"{\beta_{c_1, c_2}}"] & c_2 \otimes c_1 \arrow[u,"{h_2 \otimes h_1}"]
\end{tikzcd}
\]
commutes.
The lower part commutes because the braiding is a natural isomorphism and the upper part commutes because $d$ is braided, using the braiding on $\mathcal{C}^{\op}$ defined above.
The second statement is clear.
\end{proof}

\begin{theorem}
\label{th:braidedjantheorem}
\begin{itemize}
    \item The $2$-category of braided monoidal dagger categories is equivalent to the $2$-category of braided monoidal anti-involutive categories equipped with monoidal positivity structure:
    \[
    \Cat^\dagger_{\mathbb{E}_2} \simeq (\aICat_{\mathbb{E}_2})_P.
    \]
    \item The $2$-category of symmetric monoidal dagger categories is equivalent to the $2$-category of symmetric monoidal anti-involutive categories equipped with monoidal positivity structure:
    \[
    \Cat^\dagger_{\mathbb{E}_\infty} \simeq (\aICat_{\mathbb{E}_\infty})_P.
    \]
\end{itemize} 
\end{theorem}
\begin{proof}
Note that a braided monoidal dagger functor is simply a monoidal dagger functor which is additionally braided.
Therefore
    \[
    \Herm\colon  \aICat_{\mathbb{E}_2} \to \Cat^\dagger_{\mathbb{E}_2}
    \]
    and
    \[
    \Herm\colon  \aICat_{\mathbb{E}_\infty} \to \Cat^\dagger_{\mathbb{E}_\infty},
    \]
which we defined on objects in Lemma \ref{lem:braidedherm}, is well-defined on $1$- and $2$-morphisms.
The fact that these define $2$-functors now follows from the monoidal case.
Following the line of proof of Theorem \ref{th:monoidaljanth}, we see that the only thing left to show is that the unit and counit are braided, which is obvious.
\end{proof}

\section{Sign conventions for super Hilbert spaces}
\label{app:boringsign}

In this appendix, we will elaborate on and justify our definition of the category of super Hilbert spaces from a more categorical perspective.
Along the way, we will compare with some other possible sign conventions, highlighting the more common and equivalent definition of $\Z/2$-graded Hilbert space in the literature, as well as some nonequivalent definitions.

We start by motivating our choices of signs in Section \ref{sec:shilb} from a categorical perspective.
Note that there is an obvious way to make the assignment $V \mapsto \ol{V}$ into a symmetric monoidal functor $\sVect \to \sVect$.
This is a $\Z/2$-action in the sense that there is a canonical monoidal natural isomorphism $\lambda_V\colon  V \cong \ol{\ol{V}}$ such that $\lambda_{\ol{V}} = \ol{\lambda_V}$.
We will consider other possible sign choices for this $\Z/2$-action towards the end of this section, for example taking the monoidal data to be
\[
\ol{V \otimes W} \cong \ol{V} \otimes \ol{W} \quad \ol{v \otimes w} \mapsto (-1)^{|v||w|} \ol{v} \otimes \ol{w},
\]
or by composing $\lambda_V$ with $(-1)^F_V$.
However, we would argue that the $\Z/2$-action $\lambda$ we chose is canonical.

The functor $V \mapsto V^*$ we provided in Section \ref{sec:shilb} is a dual functor: the evaluation map $V^* \otimes V \to \C$ is given by evaluating functionals.
Since $V$ is finite-dimensional, the triangle identities hold for the coevaluation map $\C \to V \otimes V^*$ which sends $1$ to the finite sum
\[
\sum_i \epsilon^i \otimes e_i,
\]
where $\{e_i\}$ is a basis of $V$ with dual basis $\{\epsilon^i\}$.
It can be verified that the isomorphism $\Phi_V$ in the previous section is the unique isomorphism witnessing the fact that $V$ and $V^{**}$ are both duals of $V^*$.
Note that this uses the braiding on $\sVect$.
Also, two choices of dual functor are symmetric monoidally naturally isomorphic, so our choice of dual can be made without loss of generality.

Using the dual functor, there is a canonical way to make a $\Z/2$-action into an anti-involution, see \cite[Section 2.2]{mythesis} for details.
Explicitly, the functor $d\colon  \sVect \to \sVect^{\op}$ is given by $dV := \ol{V}^*$ and $\eta\colon  \id_{\sVect} \Rightarrow d^2$ is given by the composition
\begin{equation}
\label{eq:etaformula}
V \xrightarrow{\lambda_V} \ol{\ol{V}} \xrightarrow{\Phi_{\ol{\ol{V}}}} \ol{\ol{V}}^{**} \to \ol{\ol{V}^*}^* = d^2 V.
\end{equation}
The unlabeled map is the isomorphism witnessing the fact that the symmetric monoidal functor $V \mapsto \ol{V}$ preserves duals.
The monoidal data of $d$ is given by the string of isomorphisms
    \[
    \ol{V \otimes W}^* \cong (\ol{V} \otimes \ol{W})^* \cong \ol{W}^* \otimes \ol{V}^* \cong \ol{V}^* \otimes \ol{W}^*.
    \]
Note that both this isomorphism as well as $\eta$ depend on the braiding of $\sVect$.
Explicitly working out these isomorphisms yields our conventions for $d$ from the previous section.

In the literature, a \emph{$\Z/2$-graded Hilbert space} is more commonly defined without the Koszul sign present in Equation \eqref{eq:shermpairing}.
One can fit such $\Z/2$-graded Hilbert spaces into our framework of Hermitian completions as follows.
Consider the modification of the canonical symmetric monoidal anti-involution $(d,\eta)$ on $\sVect$ defined the last section by both
\begin{enumerate}
    \item changing $\eta_V$ to $\eta'_V := \eta_V \circ (-1)^F_V$;
    \item changing the monoidal data $\chi\colon  d(V \otimes W) \cong dV \otimes dW$ of the functor $d$ to the monoidal data $\chi'$ obtained by composing with
    \[
    v \otimes w \mapsto (-1)^{|v||w|} v \otimes w.
    \]
\end{enumerate}
We will consider other possible modifications in Remark \ref{rem:wrongsignconvention}.
Analogously to Proposition \ref{prop:sHerm}, one can then show:

\begin{proposition}
    The Hermitian completion 
    \[
    \Herm(\sVect, d, \eta',\chi')
    \]
    has objects super vector spaces $V = V_0 \oplus V_1$ equipped with a nondegenerate sesquilinear pairing $\langle .,. \rangle\colon  V 
\times V \to \C$ such that $V_0$ and $V_1$ are orthogonal and
\[
\langle v,w \rangle = \ol{\langle w,v \rangle},
\]
for all $v,w \in V$.
The dagger of a degree-preserving linear map is the unique operator $T^\dagger\colon  W \to V$ such that
    \[
\langle Tv, w \rangle_{W} = \langle v, T^\dagger w \rangle_{V},
\]
for all $v \in V$ and $w \in W$.
The tensor product of Hermitian pairings $\langle .,. \rangle_V$ and $\langle .,. \rangle_W$ on super vector spaces $V$ and $W$ is the Hermitian pairing 
    \[
    \langle v_1 \otimes w_1, v_2 \otimes w_2 \rangle := \langle v_1, v_2 \rangle_V \langle w_1, w_2 \rangle_W.
    \]
\end{proposition}

We then define the symmetric monoidal dagger category of $\Z/2$-graded Hilbert spaces $\Hilb_{\Z/2}$ as the full dagger subcategory on those inner products which satisfy 
\[
\langle v,v \rangle \geq 0,
\]
for all $v \in V$.
It turns out that this symmetric monoidal dagger category is equivalent to $\sHilb$.\footnote{This result is well-know to experts, also see the exercises in \cite[Section 12.5]{moore2014quantum}}

\begin{theorem}
\label{th:sHilbvsZ2grhilb}
    There is a symmetric monoidal $\dagger$-equivalence $\sHilb \simeq \Hilb_{\Z/2}$.
\end{theorem}

Instead of giving the proof, we show the following generalization of this equivalence.

\begin{proposition}
\label{prop:ungrsignconvention}
    Let $(\mathcal{C},d,\chi,\eta)$ be a symmetric monoidal anti-involutive category with monoidal positivity structure $P$.
    Let $(-1)^F$ denote a monoidal anti-involutive $B\Z/2$-action, which refines to a (not necessarily monoidal) anti-involutive $B\Z/4$-action, suggestively denoted $i^F$.
    Let $\eta' = \eta \circ (-1)^F$ and let $\chi'$ denote the monoidal 
    data 
    \[
\chi_{x,y}' := i^F_{dx \otimes dy} \circ (i^F_{dx} \otimes i^F_{dy})^{-1} \circ \chi_{x,y}.
\]
    Let $P_{i^F}$ be the collection of compositions $c \xrightarrow{i^F_c} c \xrightarrow{h} dc$ for $h \in P$.
Then the symmetric monoidal anti-involutive category $(\mathcal{C}, d, \chi, \eta)$ with positivity structure $P$ is symmetric monoidally $\dagger$-equivalent to $(\mathcal{C},d, \eta' , \chi')$ with positivity structure $P_{i^F}$.
\end{proposition}
\begin{proof}
We will first show that $P_{i^F}$ is a monoidal positivity structure, so that it actually defines a symmetric monoidal dagger category.
    For this we need to show that if $h\colon c \to dc$ is a Hermitian pairing for $(\mathcal{C}, d, \eta)$, then $h \circ i_c^F$ is a Hermitian pairing for $(\mathcal{C}, d, \eta \circ (-1)^F)$.
Indeed, by definition of the action being anti-involutive we have
\[
d(i^F_c) = (i^F_{dc})^{-1} = (-1)^F_{dc} \circ i^F_{dc}
\]
and so
\[
d(h \circ i^F_c) = d(i^F_c) \circ dh = (-1)^F_{dc} \circ i^F_{dc} \circ h \circ \eta_c^{-1}.
\]
Using naturality of the $B\Z/2$-action, we obtain the desired.
Monoidality follows immediately by the choice of $\chi'$ and the fact that $P$ is monoidal.

By Theorem \ref{th:symmetricjantheorem}, we finish the proof by finding a symmetric monoidal anti-involutive equivalence
\[
(\mathcal{C}, d, \eta, \chi) \simeq (\mathcal{C}, d, \eta', \chi'),
\]
which induces $(c,h) \mapsto (c, h \circ i_c^F)$ on the Hermitian completions
\[
\Herm (\mathcal{C},d,\chi,\eta) \to  \Herm (\mathcal{C},d,\chi',\eta').
\]
Note that the pair of $\id_\mathcal{C}$ and the monoidal natural transformation $\phi_x := di^F_x\colon  dx \to dx$ defines a symmetric monoidal anti-involutive functor
\[
F\colon  (\mathcal{C}, d, \eta, \chi) \to (\mathcal{C}, d, \eta', \chi').
\]
Indeed, the diagram saying that $F$ is anti-involutive is
\[
\begin{tikzcd}
    x \ar[d,"(-1)^F_x"] \ar[r, "\eta_x"] & d^2 x \ar[dd,"di^F_{dx}"]
    \\
     x \ar[d,"\eta_x"] & 
     \\
     d^2 x \ar[r,"d^2 i^F_x"] & d^2 x
\end{tikzcd}
\]
This follows from the fact that $i^F$ is anti-involutive, $(i^F)^2 = (-1)^F$ and $(-1)^F$ is natural.
The functor is monoidally anti-involutive by construction of $\chi'$.
The identity is a symmetric monoidal equivalence and so $F$ is an equivalence of symmetric monoidal anti-involutive categories.
This finishes the proof.
\end{proof}

The conditions of Proposition \ref{prop:ungrsignconvention} are satisfied for $\sVect$ with its standard anti-involution and $B\Z/2$-action with $i^F_V$ defined by multiplying by $i$ on the odd part of $V$.
Also note that in this example $i^F$ is not monoidal.
For example, the diagram
\[
\begin{tikzcd}
\Pi \C \otimes \Pi \C \arrow[d,"i^F \otimes i^F = -\id"] \arrow[r,"\cong"] & \C \arrow[d,"\id"]
\\
\Pi \C \otimes \Pi \C \arrow[r,"\cong"]& \C
\end{tikzcd}
\]
does not commute.
More generally, the automorphism 
\[
i^F_{V_1 \otimes V_2} \circ (i^F_{V_1} \otimes i^F_{V_2})^{-1}
\]
of $V_1 \otimes V_2$ measuring the 
failure of $i^F$ being monoidal is exactly given by $v \otimes w \mapsto (-1)^{|v||w|} v \otimes w$.
This yields a proof of Theorem \ref{th:sHilbvsZ2grhilb}

\begin{remark}
    In the setting of Proposition \ref{prop:ungrsignconvention}, there is an inverse $B\Z/4$-action we will denote $(i^{-1})^F$.
    We also obtain $\mathcal{C}_{P_{(i^{-1})^F}} \simeq \mathcal{C}_P$.
The dagger categories $\mathcal{C}_{P_{(i^{-1})^F}}$ and $\mathcal{C}_{P_{i^F}}$ have their set of Hermitian pairings related by $(-1)^F$ and so are different if $P \neq P_{(-1)^F}$.
In particular there is no equivalence of symmetric monoidal dagger categories $\mathcal{C}_{P_{(i^{-1})^F}} \simeq \mathcal{C}_{P_{i^F}}$ commuting with the anti-involutive functors to $\mathcal{C}$, even though they are always abstractly equivalent, see Lemma \ref{lem:compareposstrs}.
\end{remark}

\begin{example}
    Note that Proposition \ref{prop:ungrsignconvention} does not hold for $\sVect_\R$.
    Firstly, the $B\Z/2$-action does not have a square root, even though we can still talk about modifying the monoidal structure of $d$ by 
    $v \otimes w \mapsto (-1)^{|v||w|} v \otimes w$.
    Additionally, modifying $\eta$ for the anti-involution $d$ on $\sVect_\R$ given by $V \mapsto V^*$ by $(-1)^F$ gives a non-equivalent anti-involutive category; in one of the two all objects admit Hermitian pairings, but in the other one some do not.
    To see this concretely, note that an odd-dimensional real vector space in odd degree admits an Hilbert space structure in the ungraded sense, but not in the graded sense, as it would give a symplectic form.
\end{example}

    \begin{example}
\label{ex:ungradedshilbdual}
    It is very tempting to think that $\Hilb^{\Z/2}$ is dagger compact, i.e. the dual Hermitian pairing does have the same signature as the original.
    Indeed, it is not hard to check that if $h\colon  V \to dV$ is a Hermitian pairing in the usual convention
    \[
    \langle v,w \rangle = \ol{\langle w,v \rangle},
    \]
    then the Hermitian pairing
    \[
    V^* \xrightarrow{h^{*-1}} (dV)^* \cong d(V^*)
    \]
    has the same signature as $h$ if we take the obvious isomorphism $(dV)^* \cong d(V^*)$.
    However, in defining $\Hilb^{\Z/2}$, the monoidal data $\chi$ of $d$ is changed into $\chi'$.
    As a result, the canonical isomorphism $d(V^*) \cong (dV)^*$ saying that $d$ preserves duals is changed by $(-1)^F$.
    Since this isomorphism is used in Definition \ref{def:dualhermstr}, the dual Hermitian pairing also has its signature reversed on the odd part in this convention.
    Note that this had to be the case by Theorem \ref{th:sHilbvsZ2grhilb}.
\end{example}

\begin{remark}
\label{rem:wrongsignconvention}
    In this section, we focused on modifying the anti-involution of the previous section in a way that is convenient for describing the category of $\Z/2$-graded Hilbert spaces.
    We then showed this category is equivalent to the super Hilbert spaces described in the previous section.
    We now briefly discuss sign conventions that do not recover the correct symmetric monoidal dagger category of super Hilbert spaces.
    
    Firstly, we mention the option of changing $\eta$ to $\eta'$, but keeping $\chi$.
    Then Hermitian pairings will satisfy 
    \[
    \langle v,w \rangle = \ol{\langle w,v \rangle},
    \]
    but their tensor product is still given by
    \[
    \langle v_1 \otimes v_2, w_1 \otimes w_2 \rangle := (-1)^{|v_2||w_1|} \langle v_1, v_2 \rangle_V \langle w_1, w_2 \rangle_W.
    \]
    As a consequence, the tensor square of any odd line will always have negative definite inner product.
    We see that with this anti-involution, there is no reasonable definition of Hilbert space closed under tensor product, as this anti-involution admits no minimal monoidal positivity structure.
    Note that if we were to decide to modify $\chi$ and not $\eta$, we arrive at a similar problem.
    In fact, the resulting categories are equivalent by Proposition \ref{prop:ungrsignconvention}.
\end{remark}

\begin{remark}
We have now explored several options for changing the canonical anti-involution on $\sVect$.
The reader might instead prefer to modify the $\Z/2$-action $V \mapsto \ol{V}$ by changing the isomorphisms
\[
\lambda_V\colon  V \cong \ol{\ol{V}} \quad \mu_{V,W}\colon  \ol{V \otimes W} \cong \ol{V} \otimes \ol{W}
\]
and see how they affect the induced anti-involution $dV = \ol{V}^*$.
Overall this translation is straightfoward.
However, we note that the formula \eqref{eq:etaformula} for $\eta$ also changes by a $(-1)^F$ if we change $\mu$ into
\[
\mu'_{V,W}(v \otimes w) = (-1)^{|v||w|} \mu_{V,W}(v \otimes w).
\]
For example, we see that to obtain $\Hilb^{\Z/2}$, we only need to change $\lambda$ to $\lambda_V' = \lambda_V \circ (-1)^F_V$.
If we would instead change $\chi$ to $\chi'$, we would obtain a problematic anti-involution, as explained in Remark \ref{rem:wrongsignconvention}.
\end{remark}

\begin{remark}
    Even though from a categorical perspective the definition of $\Z/2$-graded Hilbert spaces is less well-motivated than super Hilbert spaces, the category does have the conceptual advantage that it does not require the explicit choice of a square root of $-1$ to say what an odd degree Hilbert space is.
\end{remark}

%% file: ArxivFinal.bbl
\begin{thebibliography}{10}

\bibitem{abramsky2004categorical}
Samson Abramsky and Bob Coecke.
\newblock A categorical semantics of quantum protocols.
\newblock In {\em Proceedings of the 19th Annual IEEE Symposium on Logic in Computer Science, 2004}, pages 415--425. IEEE, 2004.

\bibitem{atiyahTFT}
Michael Atiyah.
\newblock Topological quantum field theory.
\newblock {\em Publications Math{\'e}matiques de l'IH{\'E}S}, 68:175--186, 1988.

\bibitem{baezdagger}
John Baez.
\newblock Quantum quandaries: a category-theoretic perspective.
\newblock {\em The structural foundations of quantum gravity}, pages 240--265, 2006.

\bibitem{baezspans}
John Baez.
\newblock Spans in quantum theory, 2007.
\newblock Slides of the talk can be found at \url{https://math.ucr.edu/home/baez/span/}.

\bibitem{baezspinstatistics}
John Baez.
\newblock Spin, statistics, {CPT} and all that jazz, 2012.
\newblock Blog post \url{https://math.ucr.edu/home/baez/spin_stat.html}.

\bibitem{burgoynespinstatistics}
N~Burgoyne.
\newblock On the connection of spin with statistics.
\newblock {\em Il Nuovo Cimento (1955-1965)}, 8:607--609, 1958.

\bibitem{cockett2021dagger}
Robin Cockett, Cole Comfort, and Priyaa Srinivasan.
\newblock Dagger linear logic for categorical quantum mechanics.
\newblock {\em Logical Methods in Computer Science}, 17, 2021.

\bibitem{deligne1999quantum}
Pierre Deligne, Pavel Etingof, and Daniel Freed.
\newblock {\em Quantum Fields and Strings: Volume 1}.
\newblock American Mathematical Society, 1999.

\bibitem{matthewdagger}
Matthew Di~Meglio and Chris Heunen.
\newblock Dagger categories and the complex numbers: Axioms for the category of finite-dimensional {H}ilbert spaces and linear contractions.
\newblock {\em arXiv preprint arXiv:2401.06584}, 2024.

\bibitem{duck1998toward}
Ian Duck and Ennackal Chandy~George Sudarshan.
\newblock Toward an understanding of the spin-statistics theorem.
\newblock {\em American Journal of Physics}, 66(4):284--303, 1998.

\bibitem{higherdagger}
Giovanni Ferrer, Brett Hungar, Theo Johnson-Freyd, Cameron Krulewski, Lukas Müller, Nivedita, David Penneys, David Reutter, Claudia Scheimbauer, Luuk Stehouwer, and Chetan Vuppulury.
\newblock Dagger $n$-categories, 2024.
\newblock arXiv preprint arXiv:1201.2686.

\bibitem{freed2019lectures}
Daniel Freed.
\newblock {\em Lectures on field theory and topology}, volume 133.
\newblock American Mathematical Society, 2019.

\bibitem{freedhopkins}
Daniel Freed and Michael Hopkins.
\newblock Reflection positivity and invertible topological phases.
\newblock {\em Geometry \& Topology}, 25(3):1165--1330, 2021.

\bibitem{freed1999five}
Daniel~S Freed.
\newblock {\em Five lectures on supersymmetry}.
\newblock American Mathematical Soc., 1999.

\bibitem{ghez1985}
P~Ghez, R~Lima, and JE~Roberts.
\newblock {W*-}categories.
\newblock {\em Pacific J. Math.}, 120(1):79--109, 1985.

\bibitem{guido1995algebraicspinstatistics}
Daniele Guido and Roberto Longo.
\newblock An algebraic spin and statistics theorem.
\newblock {\em Communications in Mathematical Physics}, 172:517--533, 1995.

\bibitem{henriques2017bicommutant}
Andr{\'e} Henriques and David Penneys.
\newblock Bicommutant categories from fusion categories.
\newblock {\em Selecta Mathematica}, 23(3):1669--1708, 2017.

\bibitem{hesse2016frobenius}
Jan Hesse, Christoph Schweigert, and Alessandro Valentino.
\newblock Frobenius algebras and homotopy fixed points of group actions on bicategories.
\newblock {\em Theory and Applications of Categories}, 32:652--681, 2017.

\bibitem{TachikawaYonekura}
Chang-Tse Hsieh, Yuji Tachikawa, and Kazuya Yonekura.
\newblock Anomaly inflow and ${p}$-form gauge theories.
\newblock {\em Communications in Mathematical Physics}, 391(2):495--608, 2022.

\bibitem{TJFspinstatistics}
Theo Johnson-Freyd.
\newblock Spin, statistics, orientations, unitarity.
\newblock {\em Algebraic \& Geometric Topology}, 17(2):917--956, 2017.

\bibitem{moore2014quantum}
Gregory Moore.
\newblock Quantum symmetries and compatible {H}amiltonians, 2014.
\newblock Lecture notes available at \url{http://www.physics. rutgers.edu/gmoore/QuantumSymmetryBook.pdf}.

\bibitem{mullerstehouwer}
Lukas M{\"u}ller and Luuk Stehouwer.
\newblock Reflection structures and spin statistics in low dimensions.
\newblock {\em Accepted for publication at Reviews in Mathematical Physics}, 2023.

\bibitem{paulispinstatistics}
Wolfgang Pauli.
\newblock The connection between spin and statistics.
\newblock {\em Physical Review}, 58(8):716, 1940.

\bibitem{penneysUDF}
David Penneys.
\newblock Unitary dual functors for unitary multitensor categories.
\newblock {\em arXiv preprint arXiv:1808.00323}, 2018.

\bibitem{santamatospinstatistics}
Enrico Santamato and Francesco De~Martini.
\newblock Proof of the spin--statistics theorem.
\newblock {\em Foundations of Physics}, 45:858--873, 2015.

\bibitem{selinger2007positive}
Peter Selinger.
\newblock Dagger compact closed categories and completely positive maps.
\newblock {\em Electronic Notes in Theoretical computer science}, 170:139--163, 2007.

\bibitem{selinger2011survey}
Peter Selinger.
\newblock A survey of graphical languages for monoidal categories.
\newblock {\em New structures for physics}, pages 289--355, 2011.

\bibitem{stehouwermorita}
Luuk Stehouwer.
\newblock Interacting {SPT} phases are not {M}orita invariant.
\newblock {\em Letters in Mathematical Physics}, 112(3):64, 2022.

\bibitem{mythesis}
Luuk Stehouwer.
\newblock {\em Unitary fermionic topological field theory}.
\newblock PhD thesis, University of Bonn, 2024.

\bibitem{jandagger}
Luuk Stehouwer and Jan Steinebrunner.
\newblock Dagger categories via anti-involutions and positivity.
\newblock {\em arXiv preprint arXiv:2304.02928}, 2023.

\bibitem{streaterwightman}
Raymond~Frederick Streater and Arthur~Strong Wightman.
\newblock {\em {PCT}, spin and statistics, and all that}, volume~52.
\newblock Princeton {U}niversity {P}ress, 2000.

\bibitem{turaev}
Vladimir Turaev and Alexis Virelizier.
\newblock {\em Monoidal categories and topological field theory}, volume 322.
\newblock Springer, 2017.

\bibitem{weinberg1995quantum}
Steven Weinberg.
\newblock {\em The quantum theory of fields}, volume~2.
\newblock Cambridge University Press, 1995.

\bibitem{yonekura2019cobordism}
Kazuya Yonekura.
\newblock On the cobordism classification of symmetry protected topological phases.
\newblock {\em Communications in Mathematical Physics}, 368(3):1121--1173, 2019.

\end{thebibliography}
